\pdfoutput=1
\documentclass[a4paper,fleqn,longmktitle]{cas-sc}

\usepackage[authoryear,longnamesfirst]{natbib}

\usepackage{hyperref}
\usepackage{url}
\usepackage{amsmath,amssymb,amsfonts,amscd,amsthm}
\usepackage{mathrsfs}
\usepackage{algorithmic}
\usepackage{graphicx}
\usepackage{textcomp}
\usepackage{xcolor}
\usepackage{pifont}
\usepackage{verbatim}
\usepackage{threeparttable}
\usepackage{tabularx}
\usepackage{booktabs}
\usepackage{rotating}
\usepackage{wrapfig}
\usepackage{comment}
\usepackage{calc}
\usepackage{multirow}
\usepackage{array}
\usepackage{multicol}
\usepackage{ragged2e}
\usepackage{pgfplots}
\pgfplotsset{compat=1.18}
\usetikzlibrary{pgfplots.groupplots}
\newcommand{\sym}[1]{\ifmmode^{#1}\else\(^{#1}\)\fi}
\usepackage{caption}
\usepackage[english]{babel}
\usepackage{enumerate}
\usepackage{tikz}
\usepackage{bm}
\usepackage{csquotes}
\usepackage[useregional]{datetime2}
\usepackage{bbm}
\usetikzlibrary{arrows.meta, positioning, decorations.pathreplacing}
\usepackage[export]{adjustbox}
\usetikzlibrary{shapes, arrows.meta, positioning, fit, backgrounds, calc, shadows}
\usetikzlibrary{shapes.geometric}
\usetikzlibrary{intersections}
\usepackage{tcolorbox}
\usepackage{etoolbox}
\usepackage{eurosym}
\usepackage{longtable}
\usepackage{enumitem}
\ExplSyntaxOn
\cs_if_exist:NF \vbox_unpack_clear:N
  { \cs_new_eq:NN \vbox_unpack_clear:N \vbox_unpack_drop:N }
\ExplSyntaxOff

\definecolor{isoBlue}{RGB}{0, 114, 178}
\definecolor{costRed}{RGB}{213, 94, 0}
\definecolor{constrGreen}{RGB}{0, 158, 115}

\definecolor{valBlue}{RGB}{0, 114, 178}
\definecolor{profGold}{RGB}{230, 159, 0}
\definecolor{stateGray}{RGB}{100, 100, 100}

\definecolor{econBlue}{RGB}{31, 119, 180}
\definecolor{econRed}{RGB}{214, 39, 40}
\definecolor{econGold}{RGB}{188, 143, 46}

\definecolor{fossilRed}{RGB}{196, 69, 54}
\definecolor{fossilRedLight}{RGB}{248, 229, 227}
\definecolor{reGreen}{RGB}{46, 125, 50}
\definecolor{reGreenLight}{RGB}{232, 245, 233}
\definecolor{primaryBlue}{RGB}{26, 54, 93}
\definecolor{primaryBlueLight}{RGB}{232, 238, 245}
\definecolor{accentOrange}{RGB}{217, 119, 6}
\definecolor{mutedGray}{RGB}{73, 80, 87}
\definecolor{surfaceGray}{RGB}{248, 249, 250}

\tikzset{
    econAxis/.style={thick, ->, >=stealth},
    supply/.style={thick, econRed},
    demand/.style={thick, econBlue},
    priceLine/.style={thin, gray, densely dashed},
}

\theoremstyle{plain}
\newtheorem{proposition}{Proposition}

\theoremstyle{definition}
\newtheorem{definition}{Definition}
\newtheorem{assumption}{Assumption}

\newcommand{\E}{\mathbb{E}}
\newcommand{\Var}{\operatorname{Var}}
\newcommand{\Cov}{\operatorname{Cov}}
\newcommand{\pos}[1]{\left(#1\right)^{+}}
\newcommand{\REC}{\mathrm{REC}}
\newcommand{\PPA}{\mathrm{PPA}}
\newcommand{\BTM}{\mathrm{BTM}}

\begin{document}
\let\WriteBookmarks\relax
\def\floatpagepagefraction{1}
\def\textpagefraction{.001}

\shorttitle{Empirical Impacts from AI-Driven Power Demand}

\shortauthors{Golden, Balasubramanian, Balasubramanian}

\title[mode = title]{Certificates without Electrons? Theory and Evidence on Impacts from AI-Driven Power Demand}

\author[1]{Dana Golden}
\cormark[1]
\ead{dana.golden@stonybrook.edu}
\credit{Conceptualization, Methodology, Software, Data Curation, Formal Analysis, Writing -- Original Draft, Writing -- Review \& Editing, Visualization}

\affiliation[1]{organization={Department of Economics, Stony Brook University},
            city={Stony Brook},
            postcode={11794},
            state={New York},
            country={USA}}

\author[2]{Aruna Balasubramanian}
\ead{arunab@cs.stonybrook.edu}
\credit{Conceptualization, Writing -- Review \& Editing, Supervision}

\affiliation[2]{organization={Department of Computer Science, Stony Brook University},
            city={Stony Brook},
            postcode={11794},
            state={New York},
            country={USA}}

\author[2]{Niranjan Balasubramanian}
\ead{niranjan@cs.stonybrook.edu}
\credit{Conceptualization, Writing -- Review \& Editing, Supervision}

\cortext[1]{Corresponding author}

\begin{abstract}
Data centers now account for 4.4\% of United States electricity demand, yet the grid-level effectiveness of the renewable energy certificates (RECs) and power purchase agreements (PPAs) hyperscalers use to claim carbon neutrality remains unclear. We develop a game-theoretic model in which a data center operator chooses among RECs, PPAs, and behind-the-meter colocation while generators make entry decisions under endogenous financing costs. The model identifies a timing wedge—the mismatch between consumption and credited renewable generation—as a central mechanism through which AI demand degrades reliability, raises prices, and increases emissions even when RECs cover 100\% of annual consumption. Colocation with storage addresses this wedge directly and induces the greatest renewable entry by eliminating generator revenue risk.
We test these predictions by exploiting the staggered release of large language models as a natural experiment, using difference-in-differences on a novel dataset linking AI activity to local grid outcomes. AI demand significantly increases fossil generation, wholesale prices (up to 25\% in treated PJM zones), and outage frequency (0.5–1 additional outages per year) near data centers, with impacts scaling in model size. Data centers with on-site generation exhibit a sign reversal in power-quality effects, consistent with the model's prediction that behind-the-meter capacity absorbs demand spikes. Counterfactual analyses show that edge inference, spatial reallocation, and colocated storage each substantially mitigate grid impacts, while REC-only strategies do not. Together, our results demonstrate that the externalities of AI to the grid are tightly coupled to procurement design and the spatial organization of data center infrastructure.
\end{abstract}

\begin{highlights}
\item First causal estimates of AI model releases on U.S.\ power grid stability using staggered difference-in-differences
\item AI activity degrades power quality equivalent to 0.5--1 additional outages per year and increases nearby fossil generation by hundreds of GWh
\item Wholesale electricity prices increase up to 25\% in treated PJM zones during AI model activity periods
\item On-site colocated generation and distributed edge inference substantially mitigate grid impacts
\end{highlights}

\begin{keywords}
Artificial intelligence \sep Power grids \sep Electricity markets \sep Reliability \sep Power quality \sep Econometrics \sep Difference-in-differences
\end{keywords}

\maketitle

%
\section{Introduction}

\begin{figure}
    \centering
    \includegraphics[width=.8\linewidth]{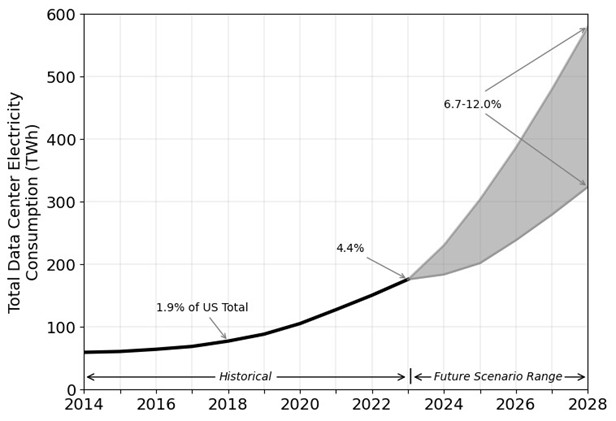}
    \caption{Projected data center growth over time. Source: DOE.}
    \label{fig:DataCenterGrowth}
    \vspace{-2pt}
\end{figure}

Over the past six years, data centers have expanded from consuming 1.9\% to 4.4\% of total U.S. electricity demand (Figure \ref{fig:DataCenterGrowth}, DOE). This rapid growth is driven largely by artificial intelligence and follows two decades of flat electricity consumption and coincides with nationwide efforts to electrify transport and heating. The combination poses a challenge: grid infrastructure was planned for stable demand, not simultaneous load growth from multiple sectors. Recent anecdotal evidence points to data centers reducing local power quality and raising prices \cite{Nicoletti_Malik_Tartar_2024}, yet systematic empirical evidence on these grid-level impacts remains scarce.

The hyperscalers driving this growth almost uniformly describe themselves as carbon-neutral, and they do so by procuring renewable attributes: unbundled renewable energy certificates (RECs) and power purchase agreements (PPAs) sized to match their annual consumption. Whether these instruments deliver the physical and grid-level outcomes they advertise is a separate question from whether they balance an annual ledger. A certificate is matched to consumption over a year; the electrons a data center actually draws are dispatched in real time, hour by hour, from whatever is on the margin. When renewable availability and around-the-clock compute load are imperfectly aligned, a firm can be ``100\% renewable'' on paper while its incremental load is served, in the hours that matter, by fossil generation. We call this gap between credited generation and consumed electrons the \emph{timing wedge}, and we argue it is the central mechanism through which AI demand raises emissions, raises prices, and degrades reliability for other grid users.

To make the wedge measurable and to organize the evidence, we develop a compact game-theoretic model (Section~\ref{sec:model}) in which a data center operator chooses among RECs, PPAs, and behind-the-meter colocation while renewable developers make entry decisions under endogenous financing costs. The model shows that the wedge survives full annual matching, that AI's spiky load amplifies it super-proportionally, and that the procurement instruments differ sharply in the renewable entry they induce: behind-the-meter colocation with storage closes the wedge and---by eliminating the generator's revenue risk---induces the most new capacity, whereas unbundled certificates induce essentially none. These results yield five testable predictions (Propositions~\ref{prop:wedge}--\ref{prop:spatial}) that we carry into the empirical analysis.

Existing work does not provide causal estimates of how AI deployments affect power systems at the {\em grid level}. Techno-economic assessments from international organizations broadly characterize rising AI workloads in the data center market \cite{RN1, RN3}. Micro-level studies examine energy performance of AI accelerators and identify workload management opportunities \cite{RN9, RN4}. A large computer science literature quantifies the compute, energy, and environmental costs of model development and deployment \cite{strubell19, schwartz2020green, RN7, RN5, morrison2025holistically}. These provide a foundation for understanding datacenter-level dynamics, but do not address the externalities experienced by other grid users—impacts on power quality, prices, and reliability that matter for energy security.

Quantifying and understanding these power grid level impacts is critical from the perspectives of energy reliability and security. Guided by the model, we seek to answer four specific empirical questions: 1) What are the effects of AI model releases on local power quality and power frequency? 2) How do AI training and inference affect electricity consumption around data centers owned by model developers? 3) How do these impacts translate into price effects in the wholesale and retail markets? Additional demand will likely increase prices, but understanding the magnitude of these price increases is important to determining the severity of the issue. 4) How will these impacts be different based on counterfactual scenarios articulating different technical evolutions of AI computing at scale? These scenarios include a shift away from cloud-based inference to edge-based computing, a geographic realignment of data centers away from population clusters, an evolution towards primarily inference-based or primarily training-based workloads, and a significant improvement in overall energy efficiency. Questions (1)--(3) test the wedge and amplification predictions (Propositions~\ref{prop:wedge} and~\ref{prop:ai-amplifies}); the counterfactuals in (4) quantify the mitigation channels of Propositions~\ref{prop:colocation}--\ref{prop:spatial}.

One way to answer the above questions is to extrapolate from energy usage measurements (micro measures) within a datacenter to system-level (macro measures) impacts on the grid. This is however, is near infeasible because it requires a great deal of specific internal information that isolates what is run within the datacenter and when. While this data with level of detail is not available (for proprietary and other practical reasons), there are data sources that provide high-level observational data about various aspects of power consumption and quality at the grid-level over the time periods when models were trained and released. \emph{How can we answer specific questions about the impact of AI models with this high-level observational data alone?}


We can draw upon econometric approaches for addressing this problem. In particular, we focus on the Difference-in-Differences (DiD) models which allows us to estimate the causal effects of AI model deployments on power systems, even when some contributing factors are unobservable. By comparing areas around data centers running AI models against control areas, and partitioning observations into pre- and post-release periods, we isolate AI-specific impacts while controlling for temporal and geographic factors. 


There are multiple technical challenges in applying this standard DiD approach in our setting: (i) There is a lack of specific information about when and where models are trained and released. Further model releases can be staggered, affecting the validity of the standard model which only uses one timepoint. (ii) The DiD method at its core relies on the parallel trends assumption i.e. we would know what would happen if the AI model were not run; in other words the estimates are valid only if we have a fairly accurate estimate of the trend of the variable of interest over time. (iii) There could be other confounding factors contributing to change in the variable of interest (e.g. extreme weather events coinciding with model release times could also cause increased power demand). (iv) No one clean dataset exists including all data useful for analysis. Often data are at different scales, lack linking tables, or require significant processing prior to use for analysis. In some cases, this requires significant efforts for manual mapping, combination, and labeling.

Our solution methodology addresses the above challenges:\\
\noindent \textbf{1) Addressing Inexact Information and Staggered Releases} 
We use a \textit{Staggered DiD} approach with multiple horizons, which has been used to overcome similar challenges in other applications ~\cite{cengiz2019effect, deshpande2019screened}. Staggered DiD with multiple horizons extends the standard two-period (pre/post), two-group setup to settings where treatment activates for different units at different times, leveraging this variation in rollout timing to assess different potential points of impact—analogous to how a staged feature rollout lets you compare early adopters against users still in the control group at each point in time.
For the question on fossil generation effects, we address confounds such as alternative fuel prices or generator efficiency issues using a two stage regression, where we first factor out the impact of these confounds, and explain the remaining demand using a second stage staggered DiD. Similarly, for estimating price effects, we need to isolate effects from other sources of power demand that can cause price changes.
We handle confounds such as weather-driven and industry production driven shifts in demand (and thus price) via a two-stage regression. 

\noindent\textbf{2) Addressing validity of the DiD method and other confounds} We conduct a wide-array of statistical robustness checks supported by integration of large amounts of data from a diverse array of sources. We run statistical tests to show that we cannot reject the null hypothesis of parallel trends. We also run placebo tests with random treatment dates and randomized locations to show that our treatment effects are not the result of randomness, and we vary treatment definitions geographically and temporally. We use time series testing for anomaly detection and structural breaks to show that the differences we find in our regressions show up in model free analysis. Together, these robustness checks provide evidence of an experiment that is valid across all relevant dimensions.

\noindent\textbf{3) Addressing data disparity and relevance issues} 
Training dates and facilities are generally unobserved, and for inference—which occurs continuously—these dates are not well-defined. We therefore rely primarily on model release dates. To strengthen identification, we supplement this with training locations, training dates, and API release dates gathered from press releases and news articles. We obtain verified training information for four models and confirmed API release dates for over half of the foundational models in our sample. We construct estimates using this most robust subset, then expand our model definitions to exploit exogenous variation for counterfactual analysis. Throughout, we use the most conservative feasible treated subset for each analysis. To construct suitable control groups, we employ propensity score matching as detailed in Appendix Subsection \ref{APP:PSM}, excluding observations without comparable pre-treatment trends.

We apply this methodology using publicly available data and proprietary on datacenter locations, power grid conditions, AI model release dates, and controls including weather and market prices. Our data creation process is described in detail in Section \ref{sectionData}. While all data pre-exist our analysis, the combination of the data creates a novel dataset that represents a significant contribution. Because the data are from government agencies, ISOs, or proprietary datasets, the data are highly trustworthy. We perform data validation to verify this. We estimate DiD models for outcomes measuring both power quality and electricity demand.

Our findings presented in Section \ref{Sec::Results} speak to the large magnitude of the impacts of AI data centers and the increasing effects of training ever-larger models over time. We find that large model releases (including the likes of GPT-3, Claude 2.1, Llama 2 documented in~\autoref{appendix:verified-models}) 
cause significant power quality deterioration in nearby areas -- \textbf{well above half the standard deviation of U.S. power-quality distributions}-- and increase nearby fossil generation on the order of terawatt-hours--- \textbf{equivalent to around 45,000 households annual consumption}-- in both training and inference phases. We also find \textbf{wholesale electricity prices increase in treated PJM zones with increases of up to 25\%} during treatment in the most affected PJM zone (Table~\ref{tab:combined_results}, Panel C).
Using these estimates, we show in Subsection \ref{sec:counterfactuals} how altering the trajectory of AI along disparate dimensions impacts the grid. We are able to provide evidence for the separate impacts of training and inference, the effects of a shift to edge computing or more concentrated training loads, and how differential efficiency improvements impact AI's effects on power usage. We also show compelling evidence that significantly increasing on-site power generation provides a promising avenue for future reductions in grid impact.

This work makes four contributions. First, we provide a tractable theoretical framework---a procurement-and-entry game built around the timing wedge---that formalizes why annual renewable matching can leave a data center's grid externalities intact and that ranks REC, PPA, and colocation strategies by their effect on renewable entry. Second, we develop a methodology for assessing macro-level grid impacts of AI models using publicly available data, even when fine-grained operational information is unavailable. Third, we provide the first causal estimates of these impacts for major frontier models, complementing the emerging literature on environmental and system-level effects of AI data centers \cite{murino2023sustainable, guidi2024environmental, thangam2024impact}; among these, the sign reversal in power-quality effects for facilities with on-site generation directly confirms the model's central mechanism. Finally, we provide policy-relevant counterfactual estimates for the impacts of the technical evolution of AI on the grid, mapping each to a mitigation channel the model identifies.

\section{Related Work}

This paper contributes to five interrelated literatures. First, techno-economic assessments document the scale of data center electricity demand: Lawrence Berkeley estimates U.S. data center consumption could reach 6.7--12\% of national electricity by 2028 \cite{shehabi2024}, while utility five-year peak demand forecasts jumped from 38 GW to 128 GW between 2023 and 2024, with approximately 90 GW attributable to data centers \cite{gridstrategies2025}. Recent simulation studies examine AI-driven load growth implications for grid planning \cite{lin2024exploding, chien2023adapting}, but these assessments characterize aggregate trends rather than isolating causal mechanisms.

Second, computer science research quantifies AI's computational and environmental costs, from training emissions \cite{strubell19, schwartz2020green} to debates over efficiency improvements \cite{RN2, RN8} and inference-time scaling \cite{kim2025cost}, including data center carbon emissions \cite{datacenter-carbon}. Related work has developed carbon intensity forecasting methods \cite{maji2022dacf, li2023gnn} and examined power forecasting for grids \cite{optimizing-grid, carbon-intensity-forecasting}, but not in the context of AI workload impacts on market outcomes. Recent empirical projections analyze the tension between AI energy growth and grid decarbonization \cite{maji2024crossroads}, while others consider grid planning for AI demand \cite{ai-grid-planning}, yet the causal effect of AI deployment on wholesale markets remains unstudied.

Third, emerging evidence points to grid-level externalities from data center concentration. Sensor data reveals spatial correlation between data center proximity and power quality degradation \cite{nicoletti2024}, PJM capacity prices increased from \$30 to \$270/MW-day between 2023--2024 amid data center load growth \cite{cmu2025}, and transmission costs are increasingly socialized across ratepayers \cite{jacobs2025}. The systems literature has examined carbon-aware geographical load shifting using locational marginal prices \cite{lindberg2021guide} and data center demand response participation \cite{chen2019datacenter, liu2014pricing, klingert2018mapping}. These studies establish correlational patterns or propose optimization frameworks but lack causal identification.

Fourth, the econometric literature on electricity markets provides methodological foundations, including difference-in-differences for demand responses \cite{roth2024did} and event studies for policy interventions \cite{pnnl2022}, but has not applied quasi-experimental methods to AI deployment impacts.

Fifth, a growing economics and policy literature questions whether voluntary clean-energy procurement delivers the physical outcomes it claims. The central concerns are \emph{additionality}---whether purchasing renewable attributes induces new generation or merely reallocates the credit for capacity that would have existed anyway---and \emph{temporal matching}---whether annual volumetric accounting, as opposed to hourly matching, reflects the emissions actually displaced. We connect these debates to the AI grid-impact question by embedding procurement choice in a model where the same annual carbon ledger can be reached by instruments with opposite consequences for entry and for the residual fossil generation the grid dispatches. To our knowledge this is the first treatment to tie the timing wedge directly to measured grid-level externalities of AI demand.

We address this gap by treating AI model releases as plausibly exogenous demand shocks, using difference-in-differences to estimate causal effects on locational marginal prices and power quality. Unlike prior work building from computational requirements or facility audits, we analyze grid-level impacts using wholesale market data, capturing realized demand effects without requiring proprietary operational information.

\section{A Model of Procurement, Entry, and the Timing Wedge}\label{sec:model}

A renewable energy certificate is a claim about an electron that has already gone wherever physics sent it. When a hyperscaler retires certificates equal to its annual consumption, it has balanced a ledger denominated in megawatt-hours per year; it has not necessarily changed which generator spun to serve its load on a windless evening. The distance between the ledger and the load is the object of this section. We formalize it as a \emph{timing wedge}, show how it survives even ``100\% renewable'' procurement, and trace how it converts AI demand into higher emissions, higher prices, and degraded reliability for other grid users.

The model serves three purposes. First, it defines the wedge precisely enough to be measured, so that the reduced-form estimates in Sections~\ref{Sec::Results}--\ref{sec:counterfactuals} acquire structural meaning rather than standing as isolated correlations. Second, it ranks the procurement instruments hyperscalers actually use---unbundled certificates, power purchase agreements, and behind-the-meter colocation---by their effect on new renewable entry, clarifying why two strategies that look identical on an annual carbon ledger can have opposite consequences for the grid. Third, it generates a small set of comparative statics (Propositions~\ref{prop:wedge}--\ref{prop:spatial}) that map one-to-one onto the empirical tests and counterfactual exercises that follow. We keep the model deliberately spare: the empirical contribution is the paper's center of gravity, and the theory earns its place by disciplining what the data are asked to show.

\subsection{Environment}\label{subsec:model-environment}

Let $(\Omega,\mathcal{F},\mathbb{P})$ be a probability space whose states $\omega$ index operating conditions (``hours'') over a representative procurement cycle. Two primitives vary across states:
\begin{itemize}
    \item \textbf{Renewable availability} $a:\Omega\to[0,1]$, the capacity factor of the region's renewable resource (e.g., the fraction of nameplate wind or solar capacity producing in state $\omega$), with mean $\bar a=\E[a]$ and $\Var(a)=\sigma_a^2>0$.
    \item \textbf{Data center load} $\ell:\Omega\to\mathbb{R}_{+}$, decomposed as $\ell=\ell_0+\Delta$, where $\ell_0$ is pre-AI baseline load and $\Delta\ge 0$ is incremental AI load. Write $\mu_\ell=\E[\ell]$ and $\sigma_\ell^2=\Var(\ell)$.
\end{itemize}
Intermittency means $\sigma_a^2>0$; the engineering of AI workloads---episodic training runs and bursty inference---means $\Delta$ is more volatile than $\ell_0$ and need not move with $a$. We summarize the alignment of supply and demand by the load--availability covariance $\Cov(\ell,a)$, which is typically weak or negative for solar-heavy regions serving round-the-clock compute.

A renewable plant of nameplate capacity $K$ (MW) produces $g(\omega)=K\,a(\omega)$. Storage with usable energy $S$ can shift generation across states within the cycle subject to an energy-balance constraint, charging when generation exceeds load and discharging otherwise; we treat storage as lossless for the baseline results and note the round-trip efficiency adjustment where it matters.

\paragraph{The residual and its three externalities.} Let $m(\omega)\ge 0$ denote clean energy \emph{delivered to and consumed by} the data center in state $\omega$ under a given procurement strategy. The load not met by clean delivery is served by the grid's marginal supplier:
\begin{equation}
    r(\omega)=\pos{\ell(\omega)-m(\omega)}.
\end{equation}
In the U.S. markets we study, the marginal unit in high net-load states is overwhelmingly fossil-fired \citep{cmu2025,PJM_DataMiner2_2025}, so the residual $r$ is what couples the data center to the grid's emissions, price, and reliability. We model three channels:
\begin{enumerate}
    \item \textbf{Emissions.} Marginal generation carries emission rate $e>0$, so expected emissions attributable to the residual are $\mathcal{E}=e\,\E[r]$.
    \item \textbf{Price.} Fossil marginal cost is increasing in dispatched quantity; over the operating range PJM's fossil supply is approximately affine \citep{PJM_DataMiner2_2025}, $P(Q)=c_0+c_1 Q$ with $c_1>0$. The residual $r$ adds to system net load $Q$, raising the market-clearing price. With supply elasticity $\varepsilon^s$, a demand shift $\Delta Q^d$ produces a percentage price change $\%\Delta p=\%\Delta Q^d/\varepsilon^s$---the relation we take to the data in Section~\ref{PriceEffects}.
    \item \textbf{Reliability.} Power-quality degradation rises in both the level and the volatility of the residual the local network must balance:
    \begin{equation}
        \Phi=\phi_1\,\E[r]+\phi_2\,\Var(r),\qquad \phi_1,\phi_2\ge 0.
        \label{eq:reliability}
    \end{equation}
    The level term captures thermal and voltage stress from sustained high net load; the variance term captures the ramps and switching transients that drive frequency excursions and harmonic distortion---the CPQI and THD outcomes of Section~\ref{subsec:power-quality-results}.
\end{enumerate}

\subsection{The Timing Wedge}\label{subsec:wedge}

\begin{definition}[Timing wedge]\label{def:wedge}
Under a procurement strategy delivering clean energy $m(\cdot)$, the \emph{timing wedge} is the expected unmatched consumption
\begin{equation}
    W \;=\; \E\!\left[\pos{\ell-m}\right]\;=\;\E[r].
\end{equation}
\emph{Volumetric (annual) matching} is the constraint $\E[m]=\mu_\ell$: certificates retired equal total consumption. \emph{Physical (hourly) matching} is the stronger requirement $m(\omega)\ge\ell(\omega)$ almost surely, under which $W=0$.
\end{definition}

The wedge makes precise the sense in which a firm can be ``100\% renewable'' on paper and fossil-dependent in fact. Decompose the residual using $\pos{x}=x+\pos{-x}$:
\begin{equation}
    W=\E\!\left[\pos{\ell-m}\right]=\underbrace{\big(\mu_\ell-\E[m]\big)}_{\text{certificate deficit}}+\underbrace{\E\!\left[\pos{m-\ell}\right]}_{\text{over-delivery}}.
    \label{eq:wedge-identity}
\end{equation}

\begin{proposition}[Certificates without electrons]\label{prop:wedge}
Suppose a strategy achieves volumetric matching, $\E[m]=\mu_\ell$. Then the physical wedge equals the expected over-delivery,
\begin{equation}
    W^{\mathrm{annual}}=\E\!\left[\pos{m-\ell}\right]\;\ge\;0,
\end{equation}
with $W^{\mathrm{annual}}>0$ whenever $\Pr(m<\ell)>0$. Consequently the residual fossil generation $\E[r]=W$, expected emissions $e\,W$, and the reliability cost $\Phi\ge\phi_1 W$ are all strictly positive even under full annual coverage, and are increasing in the misalignment between clean delivery and load.
\end{proposition}

The mechanism is conservation, not accounting fraud. A certificate balances a sum over states; it does not balance the integrand state by state. Renewable energy generated when the data center needs less than it produces (the over-delivery term) is exported or curtailed and cannot retroactively offset the fossil generation incurred in the deficit states. Annual matching zeroes the \emph{net} balance while leaving the \emph{gross} physical mismatch $\E[\pos{m-\ell}]$ untouched---and it is the gross mismatch that the grid feels. This is the formal content of ``certificates without electrons,'' and it is why the rest of the paper measures realized grid outcomes rather than contracted megawatt-hours.

\subsection{Procurement Instruments}\label{subsec:instruments}

A data center operator must serve $\ell(\omega)$ in every state and may wish to satisfy a neutrality target $\E[m]\ge\mu_\ell$. It chooses among three instruments, which differ in how they shape delivery and in how they allocate the renewable plant's revenue risk.

\paragraph{(i) Unbundled certificates (REC-only).} The operator draws grid power in real time and retires unbundled certificates equal to $\mu_\ell$. Certificates do not shape delivery: the operator's incremental load arrives on the system unshaped, so its incremental residual is $r^{\REC}=\ell$ in net-load states and
\begin{equation}
    W^{\REC}=\mu_\ell.
\end{equation}
This is the maximal wedge in Proposition~\ref{prop:wedge}: the certificate balance is zero, the physical wedge is the entire load. Unbundled certificate prices typically clear well below the levelized entry cost of new capacity and accrue largely to inframarginal, already-built renewables; we encode this stylized fact as weak additionality.

\paragraph{(ii) Power purchase agreement (PPA).} The operator contracts the entire output of a \emph{new} plant of capacity $K_{\PPA}$ at strike price $s$, settled as a contract-for-differences. The contract firms the plant's revenue and so induces entry, but the plant's output $K_{\PPA}\,a(\omega)$ is intermittent and unshaped; the operator still draws grid power in real time and the PPA settles financially. Crediting contracted generation against load state by state, delivered clean energy is $m^{\PPA}=\min(K_{\PPA}a,\ell)$ and, sizing the PPA to annual consumption ($K_{\PPA}=\mu_\ell/\bar a$ so that $\E[m^{\PPA}]\to\mu_\ell$),
\begin{equation}
    W^{\PPA}=\E\!\left[\pos{\ell-K_{\PPA}a}\right]\in\big(0,\,W^{\REC}\big).
\end{equation}
The PPA closes part of the wedge---generation that happens to coincide with load---but leaves a residual equal to the intermittency-driven shortfall, strictly positive because $\Pr(K_{\PPA}a<\ell)>0$.

\paragraph{(iii) Behind-the-meter colocation with storage (BTM).} The operator builds capacity $K_{\BTM}$ on-site, paired with storage $S$, and consumes the output behind the meter. Storage reallocates generation across states to track load. If the plant is energy-adequate ($K_{\BTM}\bar a\ge\mu_\ell$) and storage is sufficient to buffer the intra-cycle mismatch, delivery can satisfy $m^{\BTM}(\omega)\ge\ell(\omega)$ up to the storage limit, so
\begin{equation}
    W^{\BTM}\to 0.
\end{equation}
Because the load sits behind the meter, the operator's withdrawal from the grid, $\pos{\ell-m^{\BTM}}$, falls toward zero: the residual is removed from the \emph{grid's} balancing problem, and on-site capacity providing local voltage support can absorb the spikes $\Delta$ that would otherwise propagate downstream.

The instruments are ordered by the wedge they leave:
\begin{equation}
    W^{\BTM}<W^{\PPA}<W^{\REC}=\mu_\ell.
    \label{eq:wedge-order}
\end{equation}

\subsection{Entry under Endogenous Financing}\label{subsec:entry}

Renewable projects are capital-intensive and largely pre-paid; their economics turn on the cost of capital, which in project finance is governed by the certainty of contracted cash flows. We make this explicit. A competitive developer builds capacity $K$ at unit cost $\kappa$ and earns risky revenue $R_\theta(K)$ under procurement regime $\theta\in\{\REC,\PPA,\BTM\}$. Investors require a return that rises with cash-flow risk:
\begin{equation}
    \rho(\theta)=\rho_0+\lambda\,\mathrm{CV}\!\left(R_\theta\right),\qquad \lambda>0,
    \label{eq:cost-of-capital}
\end{equation}
where $\mathrm{CV}(\cdot)$ is the coefficient of variation of revenue and $\lambda$ is the price of risk---a reduced form for the standard project-finance result that merchant exposure raises spreads and lowers achievable leverage relative to contracted revenue. Per-MW expected revenue declines in aggregate entry, $\bar R_\theta(K)$ with $\bar R_\theta'(K)<0$ (value deflation as penetration rises). Free entry pins down capacity by the zero-profit condition
\begin{equation}
    \frac{\bar R_\theta(K_\theta)}{\rho(\theta)}=\kappa
    \quad\Longrightarrow\quad
    K_\theta=\bar R_\theta^{-1}\!\big(\rho(\theta)\,\kappa\big),
    \label{eq:entry}
\end{equation}
so that, holding the revenue schedule fixed, a lower required return supports more entry.

The regimes differ in how revenue risk is allocated, which we state as the model's one substantive contractual assumption.

\begin{assumption}[Risk allocation across instruments]\label{ass:risk}
Revenue risk is ordered $\mathrm{CV}(R^{\BTM})<\mathrm{CV}(R^{\PPA})<\mathrm{CV}(R^{m})$, where $R^{m}$ denotes merchant (uncontracted) revenue. Unbundled certificates do not underwrite a specific project's cash flow, so new entry built ``for'' certificates faces merchant risk; a pay-for-output PPA hedges price risk but leaves the generator bearing quantity (intermittency) risk; a behind-the-meter colocation with firm offtake transfers both price and quantity risk to the operator, which values firmed on-site power for uptime.
\end{assumption}

This ranking is the realistic one: a merchant plant lives on volatile wholesale prices times volatile output; a PPA fixes the price but not the megawatt-hours; an on-site availability or tolling arrangement fixes the payment regardless of weather. Combining \eqref{eq:cost-of-capital}--\eqref{eq:entry} with Assumption~\ref{ass:risk} yields the entry ranking.

\begin{proposition}[Procurement determines entry]\label{prop:entry}
Under Assumption~\ref{ass:risk}, induced renewable capacity is ordered
\begin{equation}
    K^{\BTM}>K^{\PPA}>K^{\REC}\approx 0.
\end{equation}
Behind-the-meter colocation induces the most entry per dollar of procurement because it minimizes the financing risk premium; unbundled certificates induce essentially none because their price accrues to inframarginal capacity and any new entry built against them carries full merchant risk.
\end{proposition}

\subsection{Equilibrium and Testable Predictions}\label{subsec:predictions}

Equilibrium is a procurement choice for the operator and an entry level for developers consistent with \eqref{eq:entry}, given availability, load, and the marginal-cost and reliability technologies. We do not require operators to internalize the residual's externalities; indeed the wedge persists precisely because private procurement (especially REC-only) and social cost diverge. The remaining three predictions follow; proofs are collected in Appendix~\ref{app:proofs}.

\begin{proposition}[AI load amplifies the wedge and its volatility]\label{prop:ai-amplifies}
Let incremental AI load $\Delta$ be added to baseline $\ell_0$ under a fixed delivery profile $m$. (i) The wedge is nondecreasing in the scale of $\Delta$. (ii) A mean-preserving spread of $\Delta$ (greater spikiness) weakly raises $W$ and strictly raises $\Var(r)$, hence the reliability cost \eqref{eq:reliability} through $\phi_2$ and the price externality through the variance of net load. (iii) Lower $\Cov(\Delta,a)$ raises $W$. Effects are therefore super-proportional in spikiness: doubling peak intensity raises grid impact by more than the increase in mean energy.
\end{proposition}

\begin{proposition}[Colocation closes the wedge and reverses the local externality]\label{prop:colocation}
With $K^{\BTM}\bar a\ge\mu_\ell$ and storage sufficient to buffer intra-cycle mismatch, $W^{\BTM}\to 0$ and the operator's grid residual $\to 0$. The local reliability contribution falls to zero and, if on-site capacity provides voltage support that absorbs $\Delta$, the sign of the local power-quality effect flips from degradation ($\Phi>0$) to neutral-or-improving ($\Phi\le 0$).
\end{proposition}

\begin{proposition}[Spatial structure of the externality]\label{prop:spatial}
(a) Distributing a given aggregate AI load $\Delta$ across $N$ nodes (an edge-inference architecture) reduces per-node spike magnitude and per-node $\Var(r)$; when $\Phi$ is convex in the local residual, total reliability cost falls. (b) Relocating load to regions with higher load--availability alignment (higher $\Cov(\ell,a)$ or greater headroom, i.e.\ a flatter $P(\cdot)$ and larger $\varepsilon^s$) lowers both the wedge $W$ and the price impact $\%\Delta p=\%\Delta Q^d/\varepsilon^s$.
\end{proposition}

\paragraph{From predictions to tests.} The five propositions partition cleanly into the paper's empirical objects. Table~\ref{tab:predictions} states the mapping; we refer back to it as each result is presented.

\begin{table}[t]
\centering
\caption{Model predictions and where they are tested.}
\label{tab:predictions}
\footnotesize
\begin{tabularx}{\linewidth}{@{}l>{\RaggedRight\arraybackslash}X l@{}}
\toprule
\textbf{Prediction} & \textbf{Empirical content} & \textbf{Evidence} \\
\midrule
P\ref{prop:wedge}: wedge persists under annual matching & Realized fossil generation, prices, and reliability respond to AI load despite procurement claims & \S\ref{subsec:fossil-fuel-demand-results}, \S\ref{subsec:price-effects-results}, \S\ref{subsec:power-quality-results} \\
P\ref{prop:entry}: procurement determines entry & Firm offtake (on-site) supports capacity that relieves the grid; certificates do not & \S\ref{sec:counterfactuals} (on-site); future direct test \\
P\ref{prop:ai-amplifies}: AI load amplifies, super-proportionally in spikiness & Demand, price, and power-quality effects scale in model size/compute & \S\ref{sec:counterfactuals} (scaling) \\
P\ref{prop:colocation}: colocation reverses the local externality & Sign reversal in power-quality effect for data centers with on-site generation & \S\ref{sec:counterfactuals} (Panel D, Table~\ref{tab:combined_results}) \\
P\ref{prop:spatial}: spatial structure & Edge inference and high-supply relocation reduce local impact & \S\ref{sec:counterfactuals} (geographic) \\
\bottomrule
\end{tabularx}
\end{table}

Two features of the mapping deserve emphasis. The marquee structural prediction, P\ref{prop:colocation}, is the one the data confirm most sharply: the model says behind-the-meter capacity should flip the sign of the local power-quality effect, and the on-site-generation results in Table~\ref{tab:combined_results} (Panel D) do exactly that. The one prediction we test only \emph{indirectly} is P\ref{prop:entry}: we observe that on-site (firmly contracted) capacity is associated with grid relief, consistent with the financing channel, but we do not separately estimate the renewable entry induced by unbundled certificates versus PPAs versus colocation. A direct additionality test---regressing new renewable and gas capacity additions on procurement type and data-center siting---is the natural next step and is, in our view, the most valuable extension of this framework. We are careful throughout not to read the entry ranking of Proposition~\ref{prop:entry} as established by the present estimates.

\section{Problem and Challenges}\label{Sec::Methods}

Our goal is to investigate three primary effects of AI data centers on power markets. We also conduct counterfactual analyses that could inform policy: assessing different scenarios of parameter size, training versus inference proportions, edge compute preference, relocation to low-population areas etc. These exercises operationalize the predictions of Section~\ref{sec:model}: the demand, price, and power-quality estimates test whether the timing wedge produces the externalities of Propositions~\ref{prop:wedge} and~\ref{prop:ai-amplifies}, and the counterfactuals quantify the mitigation channels of Propositions~\ref{prop:colocation}--\ref{prop:spatial}. The remainder of this section sets out the identification problem and the econometric design; readers interested only in the structural logic may proceed to the results, reading each estimate against Table~\ref{tab:predictions}.

\subsection{Problem introduced by AI workloads }


First, we quantify the effect of frontier AI models on power quality. 
AI workloads affect power quality through three main mechanisms. First, data centers and AI computing facilities add substantial baseload demand that can strain grid capacity and compromise voltage regulation, especially during peak periods. Second, AI workloads fluctuate rapidly with computational needs, creating unpredictable load variations that challenge stability and frequency regulation. Third, the switching power supplies and electronics essential to AI hardware generate harmonic distortion, introducing waveform distortions that spread through distribution networks and degrade power quality. See Figure \ref{fig:THD} in Appendix for an illustration of this harmonic distortion.

Second, we analyze how training and inference workloads differentially affect local power demand given that training produces sharp, episodic spikes while inference generates steadier ongoing consumption. 

Third, we estimate how wholesale electricity prices have responded to LLM development. 
AI workloads  affect wholesale electricity prices through the demand channel, with potentially indeterminate long-run retail market impacts. 
Wholesale prices are determined hourly in day-ahead markets via sealed-bid, uniform-price auctions subject to network constraints. Locational price variation reflects transmission limitations within the network. Increased demand raises wholesale prices by necessitating dispatch of higher-cost marginal generators. The wholesale price increases are passed through to retail consumers with substantial lag. Understanding wholesale demand and price shifts thus provides useful intuition regarding the direction and magnitude of retail price changes.

\subsection{The Difference-in-Differences Estimator}\label{EconometricModel}


We want to quantify the effects of running AI models on the power grid but we only have overall observational data about power demand, power quality, prices, exogenous controls, locations of data centers, and dates of model activity. We do not have access to data that specifically isolates the effects with respect to running the AI models because these require proprietary knowledge. 

To address this challenge, we can turn to Difference-in-differences (DiD), a widely-used statistical estimation technique in econometrics for causal inference with observational data. This approach has been used to answer a wide range of empirical economics questions from observational data such as telecom demand from price and quantity variation, \citet{fowlie2012industrial} evaluate environmental regulation in electricity using emissions and price data, and \citet{card1994minimum} study minimum wage effects via cross-state employment variation.

    
\begin{figure}            
    \centering
    \includegraphics[width=0.6\linewidth]{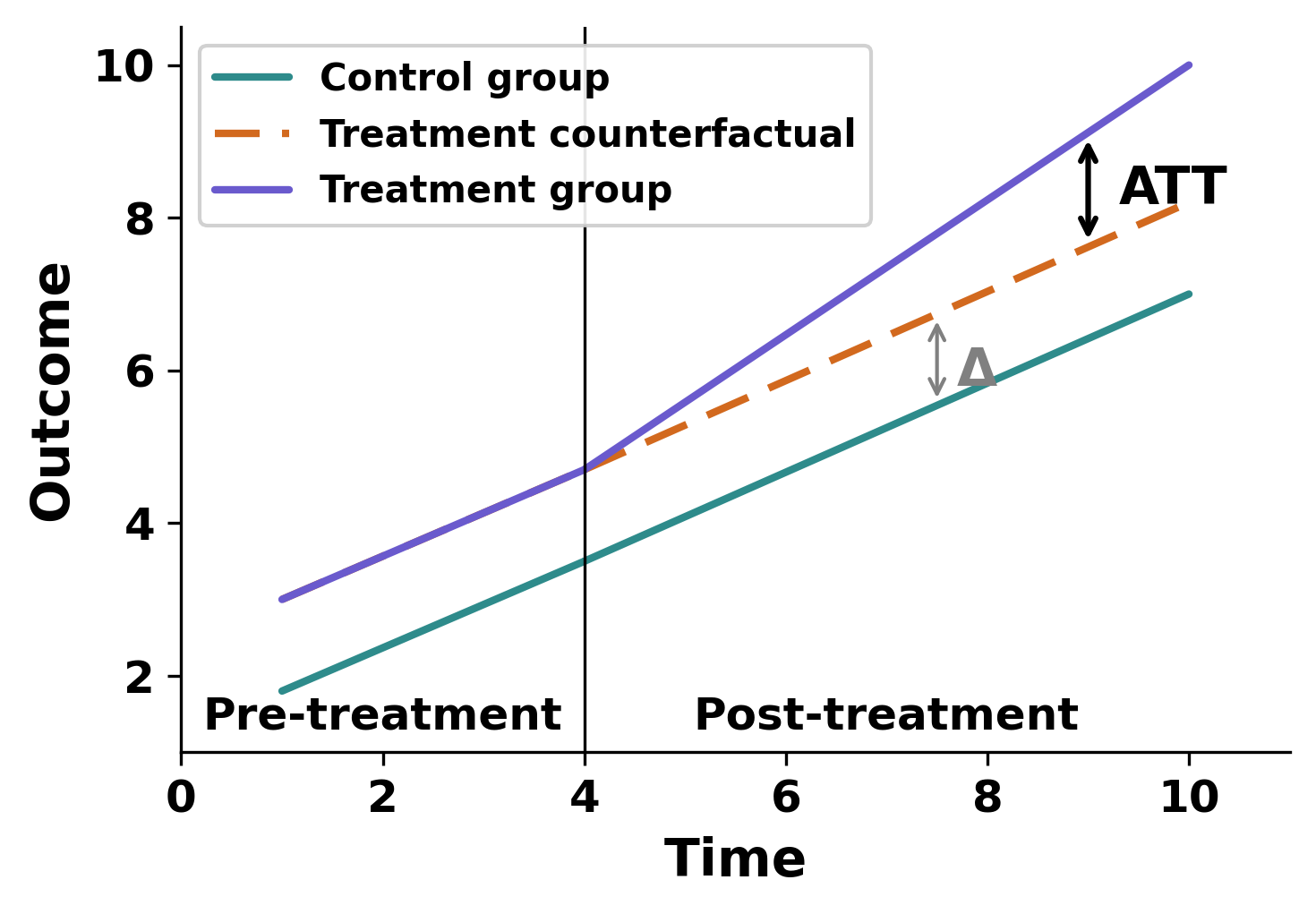}
    \caption{Difference-in-differences estimation. Under the parallel trends assumption, the counterfactual (dashed) projects treatment group outcomes absent intervention. The ATT is identified as the divergence between observed and counterfactual outcomes post-treatment.}
    \label{fig:DIDFigure}
    \vspace{-2em}
\end{figure}



The key idea behind this technique is to estimate treatment effects by comparing the change in outcomes over time for a treated group to the change for a control group, as illustrated in Figure~\ref{fig:DIDFigure}. 

For power quality, as an example, this translates to estimating the change in utility-level power quality within the geographically-defined utility area before and after the release of a major AI model. The treated group in this case would be utility areas containing data centers engaged in model training or inference during the time period of activity, and the control group would be utilities without AI activity in the same time period. The outcome variable, power quality in our example, is modeled as a function of the treatment (i.e. the running of the AI model), and other contributing factors. Formally, we can specify the power quality effect as:
\begin{equation}
    Y_{it} = \alpha + \beta \,(\text{Treatment}_i \times \text{Post}_t) + \gamma_i + \delta_t + \varepsilon_{it},
    \label{eq:DiD}
\end{equation}
where $Y_{it}$ is power quality for utility $i$ at time $t$, $\text{Treatment}_i$ indicates proximity to a data center, and $\text{Post}_t$ captures periods after the release of a major AI model. The coefficient of interest, $\beta$, measures the causal impact of model releases on local power quality. $\gamma_i$ captures unobserved factors common across all utilities, $\delta_t$ captures unobserved factors that remain fixed across time, and $\varepsilon_{it}$ is an error term that captures random disturbances or shocks. 

We can learn the co-efficients, given appropriate data about the outcome variables over time --- in this case with data about the power-quality values over time across a range of utilities containing AI data centers, and the knowledge of when and where the AI models of interest are run.

\subsection{Technical Challenges and Validating Model Assumptions}\label{DCASSUmptions}
The validity of this basic DiD approach for our problem is affected by the following issues:
\begin{enumerate}
    \item \textit{Lack of specific information about when and where models are trained and released} The model as described requires knowing specifics of the training windows, e.g., when training began, and when it was released. It is not always possible to obtain such specific information. 
    \item \textit{Staggered Events} The model also assumes a single temporal event of interest to capture pre and post treatment effects. However, models are not always released at a single time point and are often staggered.    
    \item \textit{Parallel Trends Assumption} The effects of running a model on an AI data center is estimated by the difference from what would have been in the model were not run on that data center. This assumes that we have a fair estimate of the trend that the variable of interest would have over time but for the training and inference of the AI models. 
   
    \item \textit{External Confounds} The effects we observe and attribute to AI models could be due to other confounds such as random weather events, pre-existing trends in load, or variation in fuel price.
    \item \textit{Data Engineering Challenge} No single clean dataset contains all variables needed for analysis. Data often exist at different scales, lack linking tables, or require substantial preprocessing—sometimes necessitating manual mapping, combination, and labeling. An example is the zonal mapping effort used to assign generators to zones with the zonal map shown in Figure \ref{fig:MapofZonesRegular} in the Appendix.
\end{enumerate}
    
We address the first two challenges with a stacked modification of the DiD model (\autoref{subsec:eventstudy}). We address the next two through appropriate statistical tests and robustness checks (\autoref{subsec:robustness}) and the final issue through our data creation (section \ref{sectionData}). 


\section{Methodology} 

To address the inexact information about training windows and release times, we estimate dynamic treatment effects at multiple horizons relative to model release dates. This approach, referred to as an event study or \textbf{stacked difference-in-differences} in econometrics, offers three advantages over the standard pre/post comparison discussed earlier: (1) it does not require observing the true training start date, instead using release timing as the anchor; (2) it also provides a direct test of the parallel trends assumption through examination of pre-release coefficients; and (3) it allows the data to reveal the temporal profile of effects rather than imposing an arbitrary discontinuous treatment structure.

\subsection{Stacked DiD with Multiple Horizons}
\label{subsec:eventstudy}
Formally, we define treatment exposure relative to each model's release date $r_m$ for both training and release: \textbf{1) Training windows:} The period $[r_m - \bar{\tau}_{\text{train}}, r_m)$ preceding release, during which computational resources are devoted to model training. We remain agnostic about the exact training start date; our identifying assumption is that training activity is elevated in the months before release and that any pre-existing differential trends would manifest in the earlier leads. \textbf{2) Inference window:} The period $[r_m, r_m + \bar{\tau}_{\text{infer}}]$ following release, during which the model is deployed for inference via API access or public availability.

\subsubsection{Stacked Difference-in-Difference}

Model releases occur at different calendar dates, creating a staggered adoption setting. Recent econometric literature has demonstrated that conventional two-way fixed effects (TWFE) estimators as is the case in the estimates from Equation \ref{eq:DiD}. 
can produce biased estimates under treatment effect heterogeneity 
\cite{goodman2021difference, sun2021treatment, callaway2021difference, borusyak2024revisiting}. Treatment effect heterogeneity means the policy works differently for different places or periods, which can distort standard difference-in-differences estimates. The bias arises because TWFE implicitly uses already-treated units as controls for later-treated units, and differences in treatment effects across cohorts or over time contaminate the estimated average. 

In our setting, consider regions as units and AI model releases as treatments inducing increased data center power demand. A "cohort" comprises regions whose data centers began significant AI workloads at the same time. Suppose Region A adopted AI workloads in 2022 while Region B adopted in 2023. When estimating treatment effects for Region B, TWFE implicitly uses Region A as a control—but Region A in 2023 is not untreated; its data centers have been running AI workloads for a year. If treatment effects evolve over time (e.g., power demand grows as AI infrastructure scales), comparing Region B's outcomes against Region A's already-treated outcomes introduces bias. Staggered design provides heterogeneity-robust estimators that solve for this exact issue.


To address this, we use a stacked DiD model~\cite{cengiz2019effect, deshpande2019screened}, where we construct a stacked dataset where each model release defines a separate ``sub-experiment.'' The idea here is to construct, for each treatment cohort, a separate sub-experiment using only units that are genuinely untreated at the time of that cohort's adoption. By stacking these sub-experiments—each with its own clean control group of not-yet-treated regions—we ensure that no already-treated unit ever serves as a control, thereby eliminating the bias that arises when treatment effects are heterogeneous across cohorts or time.

Formally, for each release event $m$, we create a dataset containing: (i) Treated units -- Regions linked to the releasing company's AI data centers. (ii) Control units -- Regions with no AI data center exposure during the event window, or regions whose treatment occurs sufficiently far from event $m$ that they serve as clean controls.
\begin{itemize}
    \item Time window: A symmetric window around $r_m$, e.g., $[r_m - K, r_m + K]$ for some horizon $K$.
\end{itemize}

Following \, we then stack these sub-experiments and estimate:
\begin{equation}
Y_{i,t,m} = \sum_{k=-K}^{K} \beta_k \cdot D_{ut}^k  + \delta_t + \mathbf{X}_{ut}'\theta + \varepsilon_{ut}
\label{eq:stacked}
\end{equation}
where $\gamma_{u}$ are unit-by-event fixed effects, $\delta_{u}$ are time-by-event fixed effects, $\mathbf{X}_{ut}'\theta$ represents controls, and the coefficients $\{\beta_k\}$ trace out the dynamic treatment effect at each event-time horizon $k$. Standard errors are clustered at the unit level to account for correlation across events within the same region.

\subsection{Estimating Power Quality}

Our main power quality estimates come from a two-way fixed effects (``strict'') DiD specification at the utility-month level, in which treatment enters as a single post-release indicator; given the modest cross-section of utilities, we treat the stacked event-study design below as a robustness and dynamics check rather than as the headline estimator (see Section~\ref{subsec:power-quality-results}). The event-study form of the specification is:
\begin{equation}
\text{PowerQuality}_{ut} = \sum_{k=-K}^{K} \beta_k^{PQ} \cdot D_{ut}^k + \gamma_u + \delta_t + \mathbf{X}_{ut}'\theta + \varepsilon_{ut}
\label{eq:pq_eventstudy}
\end{equation}
where:
\begin{itemize}
    \item $\text{PowerQuality}_{ut}$ is the Whisker Labs Consumer Power Quality Index for utility $u$ in month $t$;
    \item $D_{ut}^k = \mathbf{1}[\text{event-time} = k] \times \text{Treated}_u$ indicates that utility $u$ is $k$ periods from the release date of a model linked to an AI data center in its service territory;
    \item $\gamma_u$ are utility fixed effects absorbing time-invariant characteristics;
    \item $\delta_t$ are calendar month fixed effects absorbing common temporal shocks;
    \item $\mathbf{X}_{ut}$ denotes optional time-varying controls; in our baseline power-quality specification we omit these and rely on the time and geographic fixed effects alone (Section~\ref{sectionData}).
\end{itemize}

We estimate an analogous specification for total harmonic distortion (THD):
\begin{equation}
\text{THD}_{ut} = \sum_{k=-K}^{K} \beta_k^{THD} \cdot D_{ut}^k + \gamma_u + \delta_t + \mathbf{X}_{ut}'\theta + \varepsilon_{ut}
\label{eq:thd_eventstudy}
\end{equation}


\subsection{Estimating Fossil Generation}

The fossil generation analysis proceeds at the generator-hour level, exploiting finer geographic resolution. The outcome is hourly gross load at fossil-fired units from CAMPD, so this analysis identifies the fossil \emph{generation} response near AI data centers rather than total electricity demand. We maintain the stacked DiD structure while incorporating instrumental variables to address the price endogeneity concern. Price endogeneity arises because electricity prices and generation are simultaneously determined—high demand pushes prices up, while high prices call forth more generation—making it impossible to identify the causal effect of price on output from their correlation alone.

Two-stage least squares resolves this by isolating variation in price that is plausibly exogenous: in the first stage, we predict prices using an instrument that affects prices but has no direct effect on the outcome; in the second stage, we regress generator output on these predicted prices rather than observed prices.


\paragraph{First Stage.} We instrument for electricity prices using supply-side cost shifters:
\begin{equation}
p_{it} = \pi_0 + \pi_1 \text{HeatRate}_{it} + \pi_2 \text{GasPrice}_{t} + \gamma_i + \delta_t + \mathbf{X}_{it}'\pi_3 + \nu_{it}
\label{eq:firststage}
\end{equation}
where $\text{HeatRate}_{it}$ captures the thermal efficiency of generator $i$ and $\text{GasPrice}_t$ is the contemporaneous natural gas spot price. These instruments satisfy relevance (they are primary determinants of marginal generation costs) and exclusion (they affect demand only through their impact on prices, not through direct effects on AI workload scheduling or other consumption decisions).

\paragraph{Second Stage.} We estimate the stacked DiD specification using predicted prices:
\begin{equation}
\text{FossilGen}_{it} = \sum_{k=-K}^{K} \beta_k^{FD} \cdot D_{it}^k + \eta \hat{p}_{it} + \gamma_i + \delta_t + \mathbf{X}_{it}'\theta + \varepsilon_{it}
\label{eq:demand_eventstudy}
\end{equation}
where $D_{it}^k$ indicates event-time relative to model releases linked to AI data centers in generator $i$'s service area.

The coefficients $\{\beta_k^{FD}\}$ capture the dynamic effect of AI activity on fossil generation. Negative values of $k$ (pre-release) correspond to the training phase; positive values correspond to inference. The price coefficient $\eta$ controls for supply-driven price movements and provides a scaling factor for welfare calculations.



\subsection{Estimating Price Effects}\label{PriceEffects}
Beyond quantity effects, we estimate how AI-driven demand increases affect electricity prices. The key economic insight is that price impacts depend on supply flexibility. In other words, when generators can easily expand output, increases in demand (e.g. AI-driven demand) can be easily absorbed with minimal increases in price. When the demand suddenly spikes (i.e. a demand shock), the prices will increase sharply as well. This is characterized using \textit{supply elasticity} $\varepsilon^s$—the percentage change in quantity supplied per one percent price change—quantifies this responsiveness.


Our estimation proceeds in two steps. First, we use our IV estimates to quantify how AI activity shifts demand. Simply put, the average change in total fossil generation from AI demand is the DiD coefficient times the treatment indicator:
\begin{equation}
\Delta \widehat{\text{FossilDemand}}{it} = \hat{\beta} \cdot \Delta \text{AI}{it}
\label{eq:demand-change}
\end{equation}
Second, we estimate a linear supply relationship using demand controls as instruments and supply shifters as controls. While power markets involve optimal dispatch, the fossil-fuel portion of PJM’s marginal cost curve is approximately linear over the operating range we study \cite{PJM_DataMiner2_2025}.
Let $p_i$ denote the market-clearing price in market $i$, and let $Q^s$ denote
the supplied quantity. We model supply locally as linear in price,
$Q^s = a + b \cdot p_i$, where the slope parameter $b$ captures the
responsiveness of supplied quantity to price changes (i.e., how flexibly
supply can adjust to incentives). Rearranging, the implied price response to a
demand shift $\Delta AI_{it}$ is given by
$\Delta p = \hat{\beta} \cdot \frac{\Delta AI_{it}}{b}$. Expressed in percentage
terms, the price impact depends on the flexibility of supply with respect to
price:
\begin{equation}
\% \Delta p = \frac{\% \Delta Q^d}{\varepsilon^s},
\label{eq:price-impact-elasticity}
\end{equation}
where $\varepsilon^s$ denotes the percentage change in supplied quantity per
percentage change in price.


We estimate $\varepsilon^s$ via IV regression of quantity on price in log-log form, using demand-side instruments (weather-driven load variation and industrial production indices).

We focus on the PJM operator at the zonal level for this analysis. We restrict to a single ISO (operator) because differing market structures make prices incomparable across ISOs, and choose PJM because it is the largest (65 million people, 13 states), includes key data center states (Virginia, Pennsylvania), and provides 20 diverse pricing zones. Because observed training locations are outside PJM, this analysis uses the broader model set only.

\subsection{Validating Assumptions and Robustness}
\label{subsec:robustness}

We implement several diagnostic and robustness tests to assess the credibility of our estimates.

\subsubsection{Sample Selection}

Our empirical strategy balances internal validity with external relevance. Our main sample comprises four models for which we observe exact training and inference dates as well as training location, and we employ propensity score matching (described in the appendix) to support the parallel trends assumption. We then report results for progressively larger samples of models and data centers, accepting weaker identification in exchange for broader coverage. This approach enables precise, well-identified estimates from our core sample while exploring heterogeneity and counterfactuals in extended samples, with appropriate caveats regarding the latter.


\subsubsection{Pre-Trends and Parallel Trends Testing}

Our primary test of the identifying assumption examines whether pre-release coefficients $\{\beta_k : k < 0\}$ are jointly indistinguishable from zero. We report:
\begin{itemize}
    \item Point estimates and confidence intervals for each lead coefficient.
    \item An F-test of the joint null hypothesis $H_0: \beta_{-K} = \beta_{-K+1} = \cdots = \beta_{-1} = 0$.
\end{itemize}
Significant pre-trends would undermine the causal interpretation of post-release estimates.

\subsubsection{Placebo Tests with Randomized Treatment Assignment}

To verify that our estimates reflect genuine treatment effects rather than spurious correlations or specification artifacts, we conduct placebo tests with randomly assigned treatment:
\begin{enumerate}
    \item \textbf{Random treatment timing:} We hold the set of treated units fixed but randomly reassign release dates drawn from the empirical distribution of actual release dates. Under the null of no effect, placebo estimates should center on zero.
    \item \textbf{Random treatment assignment:} We hold release timing fixed but randomly reassign treatment status across units. This tests whether results are driven by the specific geographic pattern of AI data center locations.
\end{enumerate}

\subsubsection{Sensitivity to Treatment Definition}

Given uncertainty in linking models to specific training locations, we assess robustness to alternative treatment definitions in our less conservative analysis:
\begin{enumerate}
    \item \textbf{Intensive margin:} Weight treatment by the number or capacity of AI data centers in the region, rather than using a binary indicator.
    \item \textbf{Dropping ambiguous cases:} Exclude model releases where the company operates data centers in multiple balancing authorities, restricting to cases with cleaner geographic attribution.
    \item \textbf{Verified timing subset:} Estimate separately on the subset of models with externally verified training windows.
    \item \textbf{Geographic Resolution and Robustness.} The generator-level data permit more granular treatment definitions than the utility-level power quality data. We exploit this flexibility by estimating specifications under alternative treatment definitions in terms of radius.
\end{enumerate}

\subsubsection{Heterogeneity by Model Characteristics}

We examine whether treatment effects vary with observable model characteristics:
\begin{itemize}
    \item Model scale (parameter count, reported compute).
    \item Company (to assess whether effects are driven by specific firms).
    \item Release year (to assess whether effects change over time).
\end{itemize}
Heterogeneity analysis both informs mechanisms and serves as a specification check: if effects are concentrated among larger models or models from companies with substantial AI infrastructure, this lends credibility to the AI-driven interpretation.

\subsection{Counterfactual Analysis}\label{counterfactuals}

We conduct three primary sets of counterfactual exercises to investigate how alternative evolutionary paths for AI compute might affect grid-level outcomes.

\textbf{Model Scaling.} First, we examine the grid implications of continued growth in AI model size. Using the Epoch database to characterize trends in model parameters, we provide descriptive analysis of how increasing model scale translates into grid-level impacts. We then relate our estimated power demand coefficients to power quality coefficients, enabling analysis of how changes in training and inference energy efficiency propagate to power quality outcomes.

\textbf{Geographic Redistribution.} Second, we leverage the geographic locations of data centers in conjunction with our econometric estimates to conduct counterfactuals focused on spatial reallocation of AI compute. We consider two scenarios. In the first, we simulate a shift from cloud-based to edge-based inference—reallocating inference demand from concentrated hyperscaler facilities to smaller data centers distributed closer to population centers. We implement this by redistributing our estimated inference demand impacts across data centers not currently classified as AI-focused. In the second scenario, we simulate relocating training workloads from population-dense regions on the East Coast to energy-abundant regions in the center of the country. For both geographic counterfactuals, we provide sensitivity analyses in the appendix to account for potential increases in communication costs and efficiency losses.

\textbf{Colocated on-site Power Generation.} Finally, we exploit a new feature in the Aterio data, namely a flag on data centers with on-site power generation. We re-estimate both our main models separately for generators with and generators without onsite power generation. We then report the impacts to power demand and power quality of either removing all onsite generation or having onsite generation for all data centers.

Together, these counterfactuals illuminate both how the grid might be affected by different trajectories for AI compute and which levers—geographic placement, workload composition, efficiency improvements—may be available to computer scientists and policymakers seeking to mitigate grid impacts.

\section{Data}\label{sectionData}
Table \ref{tab:DataSources} summarizes the datasets we use to address each of our three research questions, distinguishing between variables of interest and the control variables described above. For variables of interest, we rely on \textbf{Aterio}, which provides comprehensive information on the location and history of data centers. The Aterio dataset documents when and where centers have been established, expanded, or retired, along with details on ownership and operational capacity (in MW). This allows us to identify data centers owned by specific hyperscalers and to quantify their scale over time. We then combine these data with external sources that provide the necessary controls. 
The data centers we use include ones that we could verifiably trace as having been the source of training of well-known large scale models such as GPT-3, PaLM, and Llama 2, among others.
\begin{table}[t!]
\scriptsize 
\centering
\caption{Data sources for each analysis question.}
\label{tab:DataSources}
\begin{threeparttable}
\setlength{\tabcolsep}{3pt} 
\renewcommand{\arraystretch}{1.05} 
\begin{tabularx}{\linewidth}{>{\raggedright\arraybackslash}p{1.1cm} X X p{2.2cm}}
\toprule
\textbf{Source} & \textbf{Content} & \textbf{Use in Analysis} & \textbf{\# Instances} \\
\midrule
\multicolumn{4}{c}{\textbf{Q1: Difference-in-Differences (Power Quality)}} \\
\midrule
Aterio & Data center locations, history, ownership, capacity (MW) & Define treatment regions; link AI model releases & 3862 data centers\\
Whisker Labs & CPQI (consumer reliability); THD (waveform distortion) & Outcome variables for DiD regressions & 2684 utility-months\\
EIA & Retail service territory maps & Define treatment/control boundaries & 72 utilities\\
\midrule
\multicolumn{4}{c}{\textbf{Q2: Instrumental Variables (Fossil Generation)}} \\
\midrule
EPA CAMPD & Hourly generator output (gross load); heat rates & Fossil generation outcome; heat rate as instrument & 22 million generator-hours\\
S\&P Capital IQ & Natural gas prices & Instrument and fuel-cost control &12 Million location-day prices \\
Aterio & Data center proximity to generators & Treatment near releasing-company centers & 3862 data centers\\
Meteostat & Temperature, precipitation & Demand and renewable controls & 22 million location-hours\\
ISOs & Zonal wholesale prices & Market controls & 54 zones\\
\midrule
\multicolumn{4}{c}{\textbf{Q3: Counterfactual Analyses (Technical Evolution of AI)}} \\
\midrule
Whisker Labs & CPQI; THD & Baseline power quality impacts for scaling projections & \\
Epoch & Model parameters, FLOPs, release dates & Scaling regressions; efficiency counterfactuals & 75 models\\
\bottomrule
\end{tabularx}
\end{threeparttable}
\end{table}


\subsection{Data and Empirical Specification}
We analyze four primary outcomes. \textbf{The Consumer Power Quality Index (CPQI)} from \cite{WhiskerLabs_THD_2024} provides a composite measure of consumer-facing reliability events (surges, sags, brownouts, interruptions), summarizing frequency and severity of power deviations at the household level. \textbf{Total Harmonic Distortion (THD)} data from the same source offer a technical measure of waveform distortion, quantifying voltage deviation from a clean 60-Hz sine wave; elevated THD indicates grid stress, reduces motor efficiency, and shortens equipment lifespan (Figure \ref{fig:THD}). Since THD data are more recent than CPQI, fewer model releases are available for THD analysis. We incorporate \textbf{generator output (gross load)} from EPA's CAMPD database, providing hourly plant-level generation spanning 2021–2023. Because CAMPD reports output at fossil-fired units, this outcome measures the local fossil generation response to data-center load rather than total electricity demand. Finally, we combine \textbf{wholesale zonal price data} from Independent System Operators (ISOs) to assess AI impacts on prices.

Our power quality regressions use a parsimonious specification with time and geographic fixed effects. Generator output regressions employ a richer specification: we instrument for endogeneity using generator heat rates and natural gas prices (S\&P Capital IQ Global) in a two-stage least squares framework, and control for weather conditions (Meteostat) to capture demand and renewable variability. EIA retail service territory files reconcile generator-, zonal-, and retail-level observations.
\subsection{Combining Data Sources}

To integrate datasets, we harmonize spatial and temporal units across sources. Figure \ref{fig:DataDiagram} demonstrates the integration at each step for our modeling purposes. Data center locations from Aterio are mapped to ISO zones, EIA territories, and counties, enabling linkage with Whisker Labs reliability indexes (CPQI and THD) and CAMPD generator output. We are the first study to integrate Aterio’s new field indicating which data centers are utilized in processing AI workloads. CPQI has 2,683 observations (mean 0.52, s.d. 0.42), while THD is more dispersed (mean 1.81, s.d. 6.77). CAMPD generator output averages 218 MW (s.d. 158), and wholesale prices average \$51/MWh (s.d. 148). Weather data capture local conditions (mean temperature 17°C, precipitation 0.12 mm), and time series are aligned at hourly or monthly resolution. External controls—natural gas prices, weather, and ISO market data—are merged by geography and time, yielding a unified panel suitable for difference-in-differences and IV regressions.

Our demand model focuses on deregulated markets in the Eastern Interconnection and ERCOT, where data on zones and prices are readily available, while our power quality model draws on Whisker Labs data from 72 retail utilities nationwide (monthly, 2022–2025). The demand regressions use hourly data for several thousand generators (2021–2023). We concentrate on hyperscaler AI companies—including Meta, Microsoft, Amazon, and Google—with Anthropic linked to Google due to its cloud partnership. Appendix materials include maps, sample descriptions, and data tables. In order to determine which data centers should be considered AI data centers and should be included in our sample, we perform a linkage procedure described in appendix Subsection \ref{subsec:treatment}.

\section{Results}\label{Sec::Results}

We report results using our stacked DiD estimator (\autoref{subsec:eventstudy}) which is designed to handle staggered AI model releases; for power quality, the two-way fixed effects specification serves as the main estimator given the smaller sample (Section~\ref{subsec:power-quality-results}). Throughout, we read each estimate against the corresponding prediction of Section~\ref{sec:model} as summarized in Table~\ref{tab:predictions}: the power-quality and demand estimates speak to the wedge and amplification predictions (Propositions~\ref{prop:wedge} and~\ref{prop:ai-amplifies}), while the on-site-generation results bear on the colocation prediction (Proposition~\ref{prop:colocation}). For all power quality and fossil generation effects, we first construct cohort-specific datasets following \citet{callaway2021difference}, for each of the treatment events corresponding to major AI model releases through 2023 (\autoref{appendix:verified-models}).
Table \ref{tab:combined_results} provides the full set of results from our main regressions including analysis of the impacts of on-site colocated generation for our counterfactuals.

\begin{table*}[t]
\centering
\caption{Summary of Estimation Results Across Specifications}
\label{tab:combined_results}
\begin{tabular}{llcp{10cm}}
\toprule
Analysis & Outcome & Coefficient & Coefficient (Effect) Interpretation \\
\midrule
\multicolumn{4}{l}{\textit{Panel A: Utility-Level Power Quality (Two-Way FE)}} \\
\addlinespace[0.5em]
Strict DiD & CPQI & 0.3795*** & Corresponds to moving from typical U.S. power quality to the bottom quartile; increases annual outages for the average customer from 1 to $\sim$1.5–2. \\
Strict DiD & THD ($>$0.8) & 0.0625*** & Causes more outages and equipment failure. \\
\addlinespace[0.5em]
\midrule
\multicolumn{4}{l}{\textit{Panel B: Generator-Level Fossil Generation (Stacked DiD)}} \\
\addlinespace[0.5em]
IV 2SLS & Pre-treatment avg & +6.33*** & \multirow{4}{7.5cm}{$>500$ additional GWh; $\sim$47,000 household-years} \\
IV 2SLS & Post-treatment avg & +5.20*** & \\
\addlinespace[0.5em]
\midrule
\multicolumn{4}{l}{\textit{Panel C: PJM AI Price Analysis}} \\
\addlinespace[0.5em]
ComEd & Supply Elasticity & 31.64 & \multirow{2}{7.5cm}{Zone-level supply elasticities and implied price impacts from the PJM supply-side analysis; the mapping from demand shifts and elasticities to price impacts follows Section~\ref{PriceEffects}} \\
ComEd & AI Price Impact & 25.39 & \\
AEP & Supply Elasticity & 8.13 & \\
AEP & AI Price Impact & 1.89 & \\
\addlinespace[0.5em]
\midrule
\multicolumn{4}{l}{\textit{Panel D: Power Quality by On-Site Generation Status}} \\
\addlinespace[0.5em]
With on-site & CPQI & $-0.122$*** & \multirow{3}{7.5cm}{On-site generation mitigates grid stress from AI load; utilities without on-site experience reliability deterioration while those with on-site see improvements} \\
Without on-site & CPQI & +0.228*** & \\
Difference & CPQI & $-0.350$*** & \\
\bottomrule
\end{tabular}
\begin{flushleft}
\footnotesize\textit{Notes:} *** $p < 0.001$, ** $p < 0.05$. Panel A: Utility and time fixed effects with cluster-robust standard errors; treatment defined as presence of AI-affiliated data center in utility service territory. Panel B: 35.6 million observations from 19 treatment events; standard errors clustered by unit-event; IV instruments: natural gas price $\times$ zone-average heat rate. OLS = Ordinary Least Squares; IV 2SLS= Instrumental Variables Two Stage Least Squares; CPQI = Consumer Power Quality Index; THD = Total Harmonic Distortion. 
\end{flushleft}
\end{table*}

\subsection{Power Quality Effects: Utility-Level Strict DiD}\label{subsec:power-quality-results}

We report the power quality effects of AI data centers in seventy one retail electric utility service territories estimated from power quality observations over 1136 utility-months. We use Consumer Power Quality Index (CPQI) and Total Harmonic Distortion (THD) as outcome variables. Our main estimates come from the two-way fixed effects (strict) DiD specification; the stacked event-study design (\autoref{eq:pq_eventstudy}) serves as a robustness check given the modest sample size.

Panel A in Table~\ref{tab:combined_results} reports results from the utility-level strict DiD specification. We use strict here to distinguish that due to the sample size, we prefer the traditional DiD as our main specification rather than the stacked alternative.
The estimated CPQI coefficient speaks to the reliability of the power supply to the consumers. The CPQI coefficient of 0.3795 indicates a power quality deterioration level that corresponds to moving from a typical U.S. power area toward the bottom quartile of power quality—roughly \textbf{the difference between $\sim$ 1 outage/year and $\sim$ 1.5–2 outages/year for the average customer}. Other reliability events such as brownouts and surges are predicted to increase by ~1-2 per year. Importantly, power surges are a known ignition source for electrical fires, contributing to the approximately 51,000 annual residential electrical fires in the U.S., with electrical arcing accounting for 74\% of ignition sources \cite{usfa2019electrical}. 

The post-treatment coefficient of 0.0625 for THD indicates that utility territories containing AI data centers experienced a 0.0625 percentage point increase in readings above the 0.8 distortion threshold following AI model releases. 
While a 0.0625 percentage point increase may appear small in isolation, it is still problematic as sustained exposure to power distortions compounds to over the multi-year lifespans of household appliances, accelerating degradation of refrigerators, air conditioners, and other motor-driven equipment \cite{NEMA_MG1,Laughner2024THD}.

\subsection{Fossil Generation Effects: Stacked DiD}
\label{subsec:fossil-fuel-demand-results}
 Fossil generation effects are estimated using hourly gross-load data from a total of 35.6 million generator-hour observations\footnote{These data are from the Clean Air Markets Program (CAMPD) dataset and include hourly generation by all US utility-scale fossil generators \cite{EPA_CAPMD_PowerSector2022}}, which includes 900,956 treated observations from generators located within 20 miles of AI-affiliated data centers and the other 34.7 million from other generators as control observations.


Average fossil generation effects are shown in the IV 2SLS rows in Panel B of Table~\ref{tab:combined_results}.\footnote{Recall that we use two stage instrumental variables (IV) regression here, where the second stage isolates the AI model caused demand-side effects from supply shocks (i.e., generator heat rate and alternate fuel source prices) modeled in the first stage. Details on the specifics of the instrumental variables specification can be found in Appendix Subsection \ref{IVDetails} and results from the first stage regression are presented in Appendix.}  Figure \ref{fig:event_study_main} shows the coefficients spread two months before (training) and after (inference) model release.


The training-phase fossil generation response corresponds to the $t=-2$ pre-treatment coefficient of $6.33$ MWh, which we find to also be statistically significant result. Note that this is the hourly generation at a single generator.  When the aggregate impact is calculated, they represent over 500 additional GWh generated nearby to a data center over the inference period in some cases. \textbf{The additional power is roughly equivalent to that of 47,000 household-years in the United States}. For inference, the $t=+2$ coefficient shows a significant generation increase of 5.20 MWh, representing the output shift at generators near data centers two months after AI model release, isolated from supply-side price variation through our instrumentation strategy.

\begin{figure}[t]
    \centering
    \includegraphics[width=0.7\linewidth]{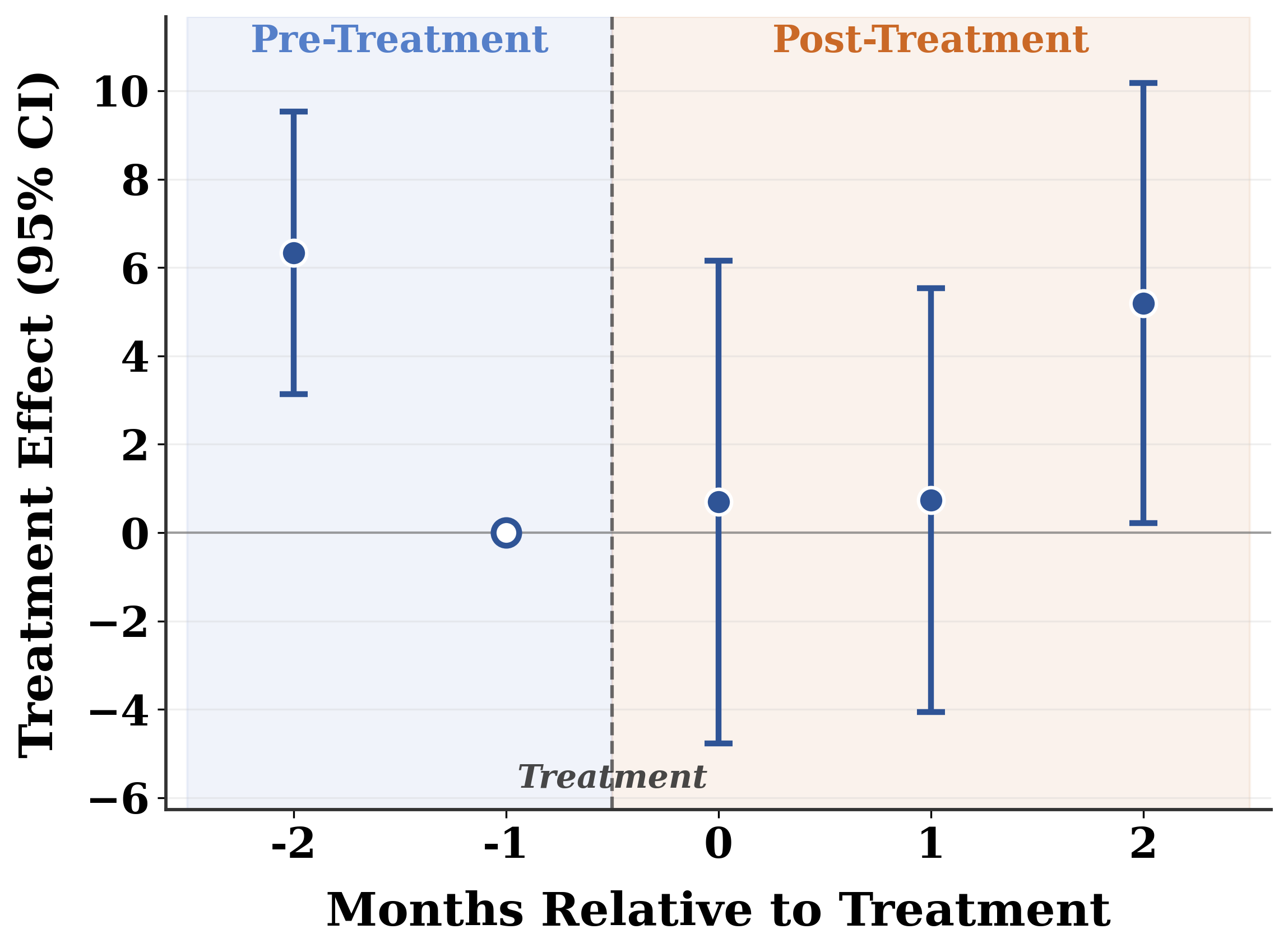}
    \caption{Stacked DiD coefficients for power demand (IV specification). Event Time 0 corresponds to AI model publication date. The pre-treatment coefficient at $t=-2$ demonstrates the training demand. Reference period is $t=-1$.}
    \label{fig:event_study_main}
    \vspace{-2em}
\end{figure}



\subsection{Price Impacts}
\label{subsec:price-effects-results}

We estimate wholesale electricity price effects by combining our demand estimates with supply curve parameters. This approach, described in Section \ref{PriceEffects}, recognizes that our IV strategy identifies demand shifts but does not directly estimate equilibrium price effects.\footnote{Our instrumental variables approach isolates exogenous variation in electricity demand driven by AI model releases. However, observed market prices reflect the intersection of supply and demand in equilibrium. To translate our demand estimates into price effects, we must incorporate information about supply-side responses, which requires additional assumptions about the shape of the supply curve.} The price impact of AI-induced demand depends on the slope of the supply curve\footnote{Recall, the supply curve describes the relationship between price and quantity supplied. Its slope reflects how much generators must be compensated to produce additional electricity—steeper slopes indicate that marginal generation costs rise quickly as output expands, typically because higher-cost units must be dispatched.}: inelastic supply, where quantity supplied responds weakly to price changes, implies large price effects, while elastic supply implies small effects. We use supply chain elasticities for PJM obtained from our supply-side instrumental variables regression (see Figure \ref{PJMPRiceElasticity} in Appendix \ref{PriceElasticityAppendix} for a map). 

Figure \ref{PJMPrice} showing price increase estimates by zone. On average AI training leads to \textbf{substantial increases in wholesale prices in PJM with the highest estimated increase surpassing 25\% in the most affected zone} (Table~\ref{tab:combined_results}, Panel C).

\begin{figure}[!htbp]
    \centering    \includegraphics[width=0.5\linewidth]{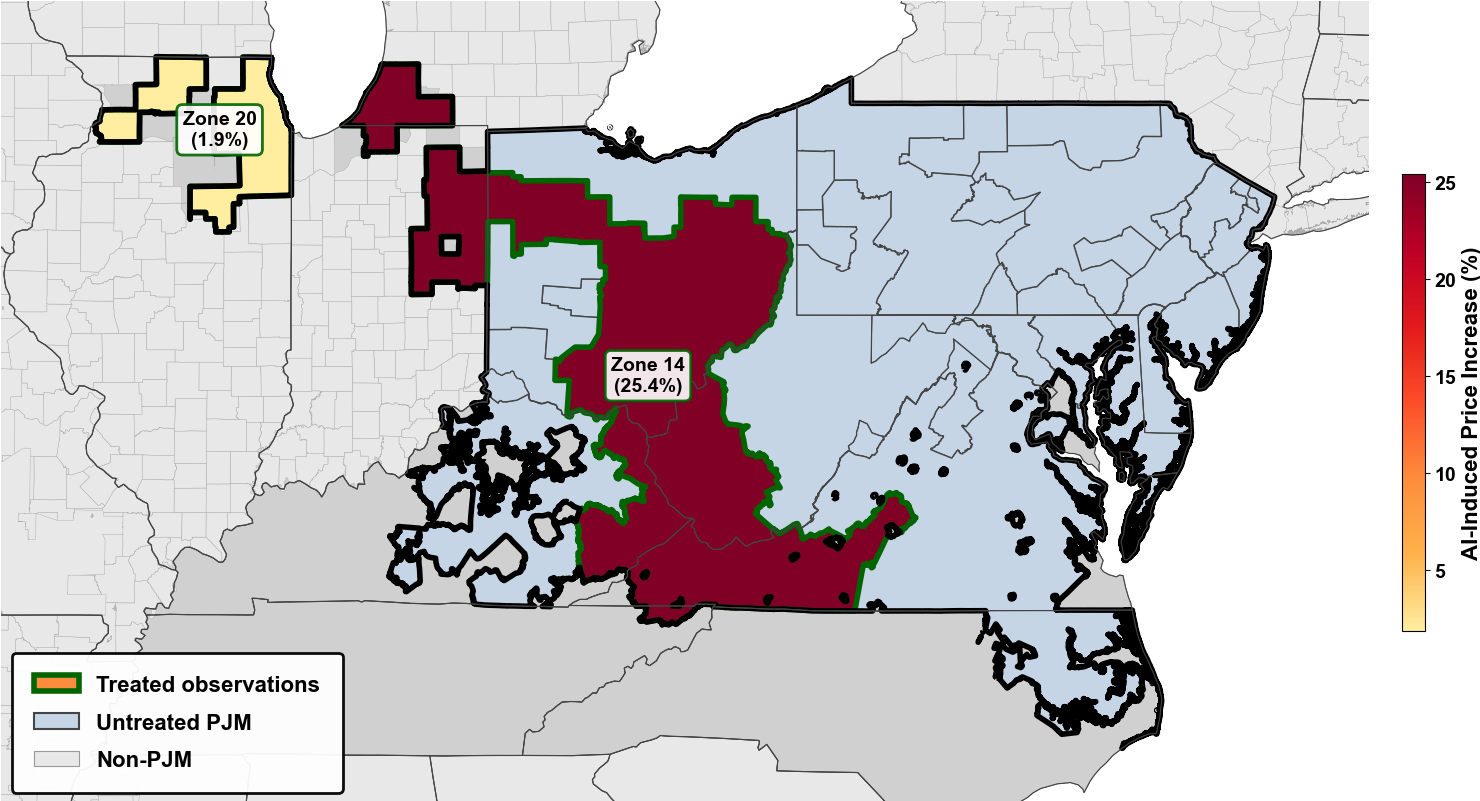}
    \caption{The map shows estimated AI-related wholesale prices increases in the PJM Interconnection in treated zones.}
    \label{PJMPrice}
    \vspace{-2em}
\end{figure}


We also estimate the expected price rise for retail customers, using estimates from the literature as documented in Appendix Subsection \ref{WholesaleRetailPassthrough}. While we are not able to provide definite evidence of the impacts on retail price pass through, it should be noted that retail electricity prices have increased overall at a rate surpassing inflation since 2020 \cite{Nicoletti_Malik_Tartar_2024,horsley2025electricbill,spp2024stateofthemarket}. Also, while increased investment in electric generation through power purchase agreements by tech firms may in the long run lead to falling retail prices, it appears that at present data center power usage is increasing wholesale prices by raising total demand, which requires dispatching higher-cost generators and thus increases the market clearing price—a price paid by all retail consumers.

 Using the most relevant wholesale-to-retail passthrough estimates for PJM of between .53 and .56 \cite{defeuilley2023pennsylvania}, \textbf{the expected consumer price rise in the most affected zone would be $\sim 13\%$}. However, using recent long-run econometric estimates of full passthrough from \cite{mackay2024wholesale}, the long-run expected rise for consumers may be closer to the full 25\%.
\section{Validity and Robustness}\label{Robustness}
We present two sets of analyses to address issues highlighted in Subsection \ref{DCASSUmptions}: formal testing to provide evidence for the parallel trends assumption, and placebo tests address potential issues related to external confounds or lack of specific knowledge about training. Additional tests are found in Appendix~\ref{App:AdditionalRobustnessChecks}.

\noindent\textbf{1) Formal Testing for Parallel Trends: } For power demand specifically, the formal joint test of pre-treatment coefficients using a Wald test of joint significance yields a chi-squared statistic of 0.99 (df = 1) with p-value = 0.3196. \textbf{We cannot reject the null hypothesis of parallel pre-trends at conventional significance levels.} This provides support for our assumption that treated and control generators followed similar demand trajectories prior to AI model releases i.e., our identification assumption is valid. Some of our samples demonstrate significant evidence against the the parallel trends test prior to propensity score matching (PSM) that we use for sample selection. However, across all matched samples for both power quality and power demand after PSM, we find \textit{no evidence of violations of the parallel trends assumption.}
This implies that the technique correctly handles the issue where it exists. 
 

\noindent\textbf{2) Placebo Tests}
We conduct two placebo exercises to demonstrate that the treatments are meaningful and not the result of random chance, we implement random date and random location placebo tests. Both yield a p-value of 1.00 for the coefficient in our fossil generation, CPQI, and THD analyses, confirming that our results do not arise from spuriously chosen locations or dates.

\section{Counterfactual Analyses}
\label{sec:counterfactuals}

We present three main counterfactual exercises to inform stakeholders and policymakers about how alternative evolutionary paths for AI compute might affect grid outcomes. These combine our econometric estimates with data on model characteristics and infrastructure configurations to explore three dimensions: model scaling, geographic redistribution, and workload composition. Each maps to a mechanism the model isolates: scaling traces the amplification of Proposition~\ref{prop:ai-amplifies}, geographic redistribution quantifies the spatial channel of Proposition~\ref{prop:spatial}, and on-site generation operationalizes the wedge-closing channel of Proposition~\ref{prop:colocation}.\\

\noindent\textbf{1) Model Scaling.} We combine data on model size (parameter counts) with our econometric estimates from our least conservative specification to extrapolate grid impacts at different scales.

\noindent\textbf{(a) Power Quality.} Fitting a trend line to power quality impacts for Llama-3.1 (405B), PaLM (540B), and Llama 4 Behemoth (2T), we find that scaling from a 2 trillion parameter model to a 4 trillion parameter model would increase power quality impact from \textbf{0.321 to 0.434}---a deterioration of roughly 0.113 units, or just over 35\% relative to the 2 trillion baseline. The relationship is exponential over log parameter counts, implying that larger models require proportionally greater parameter reductions to achieve the same percentage improvement in power quality. \textbf{(b) Power Demand.} Using an analogous extrapolation, we find that \textbf{each 1\% parameter increase leads to approximately 0.15\% higher power demand} within three months of model release. Scaling from 540 billion to one trillion parameters would \textbf{increase total power demand by 10.5\%}. For DALL-E, halving parameters would \textbf{reduce inference demand by approximately 300 GWh}---equivalent to the annual consumption of 27,000--30,000 homes.\\

\noindent\textbf{2) Training Compute.} Training compute requirements also have measurable effects: a 1\% increase in training FLOPs is associated with a \textbf{0.1\% increase in power demand}. For GPT-3.5, halving training compute would reduce energy use by enough to power 8,000--10,000 U.S.\ homes for a year.\\

\noindent\textbf{3) Colocated On-site Generation.} Data centers increasingly deploy colocated on-site generation to reduce grid dependence. Our inventory identifies 87 facilities with on-site capability and 12,401 without. We test whether on-site capacity moderates estimated effects by interacting treatment with on-site status in the utility-level strict DiD specification. Results indicate that on-site generation substantially mitigates measured grid impacts. The sign reversal for CPQI is notable and provides the most direct confirmation of the model's central mechanism (Proposition~\ref{prop:colocation}): data centers \textit{with} on-site generation are associated with improved power quality (negative coefficient), while those \textit{without} are associated with deterioration (positive coefficient). This is precisely the sign flip the model predicts when behind-the-meter capacity removes the operator's residual from the grid's balancing problem---\textbf{on-site generation absorbs demand spikes that would otherwise propagate to the grid}, closing the timing wedge for the served load and implying that aggregate grid impacts may understate total AI energy consumption.

Figure~\ref{fig:CounterfactualCombined} extends this to policy-relevant scenarios. Remarkably, universal adoption of on-site power would shift impact of data centers from a net detriment to power quality to a net improvement.\\


\noindent\textbf{4) Geographic Redistribution}


\noindent\textbf{(a) Distributed Inference (Edge Computing).} In our sample where we know AI data centers but not specific training and inference sites, we simulate a shift from cloud-based to edge-based inference, redistributing inference demand from concentrated hyperscaler facilities to smaller data centers located closer to population centers. Under this scenario, total demand added in a particular zone falls from a baseline of 162 MWh per hour over the entire treatment period to 4 MWh per hour. Figure~\ref{fig:CounterfactualCombined} illustrates this effect of shifting the composition of inference.


\noindent\textbf{(b) Concentration in Low-Population, High-Supply Regions.} Alternatively, we consider relocating both training and inference to regions with abundant electricity supply and lower baseline load, such as areas in the central United States with substantial renewable or fossil capacity. We simulate this by shifting demand from the baseline case where training and inference are assumed distributed to a case where the only data center for training and inference is the largest data center owned by a hyperscaler. This leads to an extreme concentration of added demand with nearly 100 GWh added near data centers in training locations. Figure~\ref{fig:CounterfactualCombined} shows this scenario.

This analysis has the important caveat that shifting to one data center for all training activity is extreme, with many possible configurations between full decentralization and full centralization. Centralization may also yield other benefits, such as concentration near renewable generation and reduced energy consumption from lower communication costs.


\begin{figure}
    \centering
    \includegraphics[width=.9\linewidth]{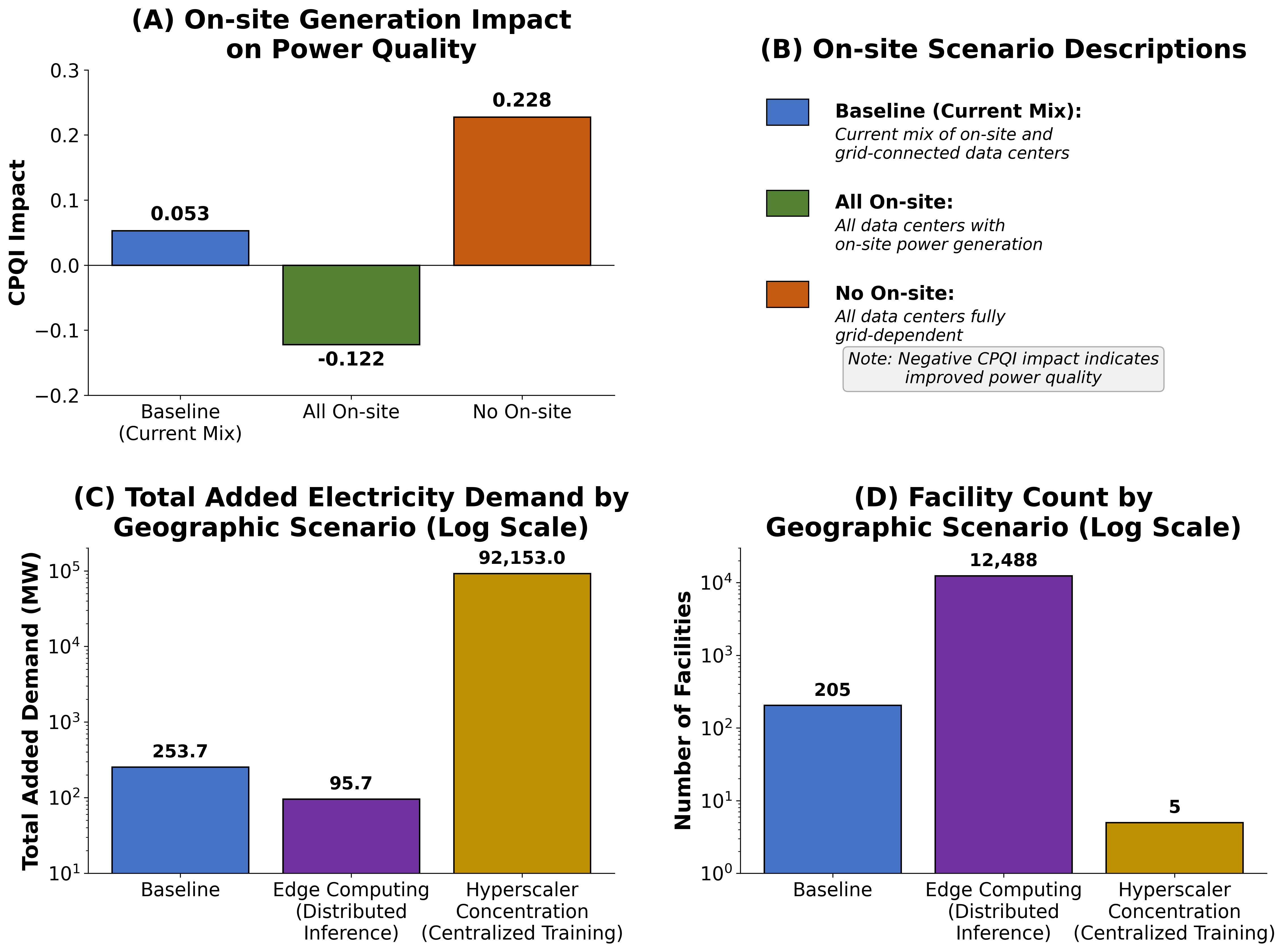}
    \caption{Effects of Geographic Concentration and On-site Generation on AI Impacts on the grid.}
    \label{fig:CounterfactualCombined}
    \vspace{-2em}
\end{figure}

\section{Conclusion}
Understanding how AI models affect U.S. power grids is central to energy reliability and security, and the question is not settled by the annual carbon ledgers hyperscalers report. We have argued and then shown that the grid consequences of AI demand turn on a timing wedge between credited renewable generation and consumed electrons, a wedge that survives full annual matching and that procurement and spatial design either widen or close. The model of Section~\ref{sec:model} makes this precise and generates five predictions; our econometric estimates, built on a novel dataset linking AI activity to local grid outcomes, bear them out. AI model releases significantly degrade local power quality and raise nearby fossil generation, wholesale prices rise by as much as 25\% in the most affected PJM zone, and impacts scale with model size---the super-proportional amplification the model attributes to spiky load. Most tellingly, data centers with on-site generation exhibit a sign reversal in power-quality effects, the exact signature the model predicts when behind-the-meter capacity removes a facility's residual from the grid. Our counterfactuals then quantify the levers the model identifies: edge inference, geographic reallocation toward high-supply regions, and colocated generation each substantially mitigate grid impacts, while certificate-only strategies do not.

The policy reading is direct. Because the externality is coupled to procurement design rather than to total consumption alone, instruments that look equivalent on an annual ledger are not equivalent for the grid: firm, co-located, storage-backed supply both closes the wedge and---by removing the generator's revenue risk---does the most to bring new clean capacity online, whereas unbundled certificates do neither. The most valuable extension of this work is therefore a direct test of the entry prediction (Proposition~\ref{prop:entry}): estimating the new renewable and natural-gas capacity additions induced by REC, PPA, and colocation strategies around data-center siting events. More broadly, the design we use combines econometric causal inference with domain-specific data to recover externalities that providers do not disclose. This design should transfer to other opaque infrastructure impacts of AI, from transmission congestion to cooling-water demand, wherever the certificate and the electron part ways.

\bibliography{Energy}

\begin{thebibliography}{71}
\expandafter\ifx\csname natexlab\endcsname\relax\def\natexlab#1{#1}\fi
\providecommand{\url}[1]{\texttt{#1}}
\providecommand{\href}[2]{#2}
\providecommand{\path}[1]{#1}
\providecommand{\DOIprefix}{doi:}
\providecommand{\ArXivprefix}{arXiv:}
\providecommand{\URLprefix}{URL: }
\providecommand{\Pubmedprefix}{pmid:}
\providecommand{\doi}[1]{\href{http://dx.doi.org/#1}{\path{#1}}}
\providecommand{\Pubmed}[1]{\href{pmid:#1}{\path{#1}}}
\providecommand{\bibinfo}[2]{#2}
\ifx\xfnm\relax \def\xfnm[#1]{\unskip,\space#1}\fi
\bibitem[{Abadie(2005)}]{abadie2005semiparametric}
\bibinfo{author}{Abadie, A.}, \bibinfo{year}{2005}.
\newblock \bibinfo{title}{Semiparametric difference-in-differences estimators}.
\newblock \bibinfo{journal}{Review of Economic Studies} \bibinfo{volume}{72},
  \bibinfo{pages}{1--19}.
\newblock \DOIprefix\doi{10.1111/0034-6527.00321}.
\bibitem[{Agency(2022)}]{EPA_CAPMD_PowerSector2022}
\bibinfo{author}{Agency, U.S.E.P.}, \bibinfo{year}{2022}.
\newblock \bibinfo{title}{Power sector emissions data (campd): Clean air
  markets program data}.
\newblock \bibinfo{howpublished}{\url{https://campd.epa.gov/}}.
\newblock \bibinfo{note}{Office of Atmospheric Protection, Clean Air Markets
  Division. Accessed on Month Day, Year}.
\bibitem[{{Anthropic}(2023)}]{anthropic_claude2_1_2023}
\bibinfo{author}{{Anthropic}}, \bibinfo{year}{2023}.
\newblock \bibinfo{title}{Claude 2.1}.
\newblock \URLprefix \url{https://www.anthropic.com/news/claude-2-1}.
\bibitem[{Austin(2011)}]{austin2011introduction}
\bibinfo{author}{Austin, P.C.}, \bibinfo{year}{2011}.
\newblock \bibinfo{title}{An introduction to propensity score methods for
  reducing the effects of confounding in observational studies}.
\newblock \bibinfo{journal}{Multivariate Behavioral Research}
  \bibinfo{volume}{46}, \bibinfo{pages}{399--424}.
\newblock \DOIprefix\doi{10.1080/00273171.2011.568786}.
\bibitem[{Berthelot et~al.(2024)Berthelot, Caron, Jay and Lef{\`e}vre}]{RN5}
\bibinfo{author}{Berthelot, A.}, \bibinfo{author}{Caron, E.},
  \bibinfo{author}{Jay, M.}, \bibinfo{author}{Lef{\`e}vre, L.},
  \bibinfo{year}{2024}.
\newblock \bibinfo{title}{Estimating the environmental impact of generative-ai
  services using an lca-based methodology}.
\newblock \bibinfo{journal}{Procedia CIRP} \bibinfo{volume}{122},
  \bibinfo{pages}{707--712}.
\bibitem[{Borusyak et~al.(2024)Borusyak, Jaravel and
  Spiess}]{borusyak2024revisiting}
\bibinfo{author}{Borusyak, K.}, \bibinfo{author}{Jaravel, X.},
  \bibinfo{author}{Spiess, J.}, \bibinfo{year}{2024}.
\newblock \bibinfo{title}{Revisiting event-study designs: robust and efficient
  estimation}.
\newblock \bibinfo{journal}{Review of Economic Studies} \bibinfo{volume}{91},
  \bibinfo{pages}{3253--3285}.
\bibitem[{Brown et~al.(2020)Brown, Tsai, Woo, Zarnikau and
  Zhu}]{brown2020texas}
\bibinfo{author}{Brown, D.P.}, \bibinfo{author}{Tsai, C.H.},
  \bibinfo{author}{Woo, C.K.}, \bibinfo{author}{Zarnikau, J.},
  \bibinfo{author}{Zhu, S.}, \bibinfo{year}{2020}.
\newblock \bibinfo{title}{Residential electricity pricing in {Texas}'s
  competitive retail market}.
\newblock \bibinfo{journal}{Energy Economics} \bibinfo{volume}{92},
  \bibinfo{pages}{104953}.
\newblock \DOIprefix\doi{10.1016/j.eneco.2020.104953}. \bibinfo{note}{estimates
  43--47\% pass-through of wholesale price forecasts to retail quotes in
  ERCOT}.
\bibitem[{Callaway and Sant’Anna(2021)}]{callaway2021difference}
\bibinfo{author}{Callaway, B.}, \bibinfo{author}{Sant’Anna, P.H.},
  \bibinfo{year}{2021}.
\newblock \bibinfo{title}{Difference-in-differences with multiple time
  periods}.
\newblock \bibinfo{journal}{Journal of econometrics} \bibinfo{volume}{225},
  \bibinfo{pages}{200--230}.
\bibitem[{Card and Krueger(1994)}]{card1994minimum}
\bibinfo{author}{Card, D.}, \bibinfo{author}{Krueger, A.B.},
  \bibinfo{year}{1994}.
\newblock \bibinfo{title}{Minimum wages and employment: A case study of the
  fast-food industry in new jersey and pennsylvania}.
\newblock \bibinfo{journal}{American Economic Review} \bibinfo{volume}{84},
  \bibinfo{pages}{772--793}.
\bibitem[{Cengiz et~al.(2019)Cengiz, Dube, Lindner and
  Zipperer}]{cengiz2019effect}
\bibinfo{author}{Cengiz, D.}, \bibinfo{author}{Dube, A.},
  \bibinfo{author}{Lindner, A.}, \bibinfo{author}{Zipperer, B.},
  \bibinfo{year}{2019}.
\newblock \bibinfo{title}{The effect of minimum wages on low-wage jobs}.
\newblock \bibinfo{journal}{The Quarterly Journal of Economics}
  \bibinfo{volume}{134}, \bibinfo{pages}{1405--1454}.
\bibitem[{Chen et~al.(2019)Chen, Zhang, Caramanis and
  Coskun}]{chen2019datacenter}
\bibinfo{author}{Chen, H.}, \bibinfo{author}{Zhang, Y.},
  \bibinfo{author}{Caramanis, M.C.}, \bibinfo{author}{Coskun, A.K.},
  \bibinfo{year}{2019}.
\newblock \bibinfo{title}{Data center participation in demand response programs
  with quality-of-service guarantees}, in: \bibinfo{booktitle}{Proceedings of
  the 10th ACM International Conference on Future Energy Systems},
  \bibinfo{publisher}{ACM}, \bibinfo{address}{New York, NY, USA}. pp.
  \bibinfo{pages}{154--164}.
\newblock \DOIprefix\doi{10.1145/3307772.3328309}.
\bibitem[{Chien et~al.(2023)Chien, Lin, Nguyen, Rao, Sharma and
  Wijayawardana}]{chien2023adapting}
\bibinfo{author}{Chien, A.A.}, \bibinfo{author}{Lin, L.},
  \bibinfo{author}{Nguyen, H.}, \bibinfo{author}{Rao, V.},
  \bibinfo{author}{Sharma, T.}, \bibinfo{author}{Wijayawardana, R.},
  \bibinfo{year}{2023}.
\newblock \bibinfo{title}{Adapting datacenter capacity for greener datacenters
  and grid}, in: \bibinfo{booktitle}{Proceedings of the 14th ACM International
  Conference on Future Energy Systems}, \bibinfo{publisher}{ACM},
  \bibinfo{address}{New York, NY, USA}.
\newblock \DOIprefix\doi{10.1145/3575813.3595197}.
\bibitem[{Crump et~al.(2009)Crump, Hotz, Imbens and Mitnik}]{crump2009dealing}
\bibinfo{author}{Crump, R.K.}, \bibinfo{author}{Hotz, V.J.},
  \bibinfo{author}{Imbens, G.W.}, \bibinfo{author}{Mitnik, O.A.},
  \bibinfo{year}{2009}.
\newblock \bibinfo{title}{Dealing with limited overlap in estimation of average
  treatment effects}.
\newblock \bibinfo{journal}{Biometrika} \bibinfo{volume}{96},
  \bibinfo{pages}{187--199}.
\newblock \DOIprefix\doi{10.1093/biomet/asn055}.
\bibitem[{De~Feuillley and Fontaine(2023)}]{defeuilley2023pennsylvania}
\bibinfo{author}{De~Feuillley, C.}, \bibinfo{author}{Fontaine, G.},
  \bibinfo{year}{2023}.
\newblock \bibinfo{title}{Pass-through in residential retail electricity
  competition: {Evidence} from {Pennsylvania}}.
\newblock \bibinfo{journal}{Utilities Policy} \bibinfo{volume}{85},
  \bibinfo{pages}{101678}.
\newblock \DOIprefix\doi{10.1016/j.jup.2023.101678}. \bibinfo{note}{reports
  52.5\% (EDC) and 56.3\% (EGS) pass-through rates for PJM Pennsylvania
  residential}.
\bibitem[{Deryugina et~al.(2020)Deryugina, MacKay and
  Reif}]{deryugina2020longrun}
\bibinfo{author}{Deryugina, T.}, \bibinfo{author}{MacKay, A.},
  \bibinfo{author}{Reif, J.}, \bibinfo{year}{2020}.
\newblock \bibinfo{title}{The Long-Run Dynamics of Electricity Demand:
  {Evidence} from Municipal Aggregation}.
\newblock \bibinfo{type}{Working Paper} \bibinfo{number}{23483}. National
  Bureau of Economic Research.
\newblock \DOIprefix\doi{10.3386/w23483}. \bibinfo{note}{uses Illinois
  municipal aggregation; finds residential elasticity of $-0.09$ (short-run) to
  $-0.35$ (long-run)}.
\bibitem[{Deshpande and Li(2019)}]{deshpande2019screened}
\bibinfo{author}{Deshpande, M.}, \bibinfo{author}{Li, Y.},
  \bibinfo{year}{2019}.
\newblock \bibinfo{title}{Who is screened out? application costs and the
  targeting of disability programs}.
\newblock \bibinfo{journal}{American Economic Journal: Economic Policy}
  \bibinfo{volume}{11}, \bibinfo{pages}{213--248}.
\bibitem[{EIA(2024)}]{USHouseholdConsumption}
\bibinfo{author}{EIA}, \bibinfo{year}{2024}.
\newblock \bibinfo{title}{How much electricity does an american home use?}
\newblock \URLprefix \url{https://www.eia.gov/tools/faqs/faq.php?id=97&t=3}.
  \bibinfo{note}{average U.S. residential electricity consumption is about
  10,791 kWh/year, approximately 1.23 kW continuous load}.
\bibitem[{EPRI(2024)}]{RN3}
\bibinfo{author}{EPRI}, \bibinfo{year}{2024}.
\newblock \bibinfo{title}{Powering Intelligence: Analyzing Artificial
  Intelligence and Data Center Energy Consumption}.
\newblock \bibinfo{type}{Report}.
\bibitem[{Fowlie et~al.(2012)Fowlie, Holland and Mansur}]{fowlie2012industrial}
\bibinfo{author}{Fowlie, M.}, \bibinfo{author}{Holland, S.P.},
  \bibinfo{author}{Mansur, E.T.}, \bibinfo{year}{2012}.
\newblock \bibinfo{title}{What do emissions markets deliver and to whom?
  evidence from southern california’s nox trading program}.
\newblock \bibinfo{journal}{American Economic Review} \bibinfo{volume}{102},
  \bibinfo{pages}{965--993}.
\bibitem[{Goggin and Gramlich(2025)}]{gridstrategies2025}
\bibinfo{author}{Goggin, M.}, \bibinfo{author}{Gramlich, R.},
  \bibinfo{year}{2025}.
\newblock \bibinfo{title}{Power Demand Forecasts Revised Up for Third Year
  Running, Led by Data Centers}.
\newblock \bibinfo{type}{Technical Report}. Grid Strategies LLC.
\newblock \URLprefix
  \url{https://gridstrategiesllc.com/wp-content/uploads/Grid-Strategies-National-Load-Growth-Report-2025.pdf}.
\bibitem[{Goodman-Bacon(2021)}]{goodman2021difference}
\bibinfo{author}{Goodman-Bacon, A.}, \bibinfo{year}{2021}.
\newblock \bibinfo{title}{Difference-in-differences with variation in treatment
  timing}.
\newblock \bibinfo{journal}{Journal of econometrics} \bibinfo{volume}{225},
  \bibinfo{pages}{254--277}.
\bibitem[{{Google Cloud}(2023)}]{googlecloud2023mlhub}
\bibinfo{author}{{Google Cloud}}, \bibinfo{year}{2023}.
\newblock \bibinfo{title}{{Google Cloud} opens its first {ML Hub}, focused on
  {AI}-optimized compute}.
\newblock
  \bibinfo{howpublished}{\url{https://cloud.google.com/blog/products/ai-machine-learning/google-cloud-launches-new-ai-ml-hub}}.
\newblock \bibinfo{note}{Accessed: 2025}.
\bibitem[{{Google Developers Blog}(2023)}]{google_palm_api_makersuite_2023}
\bibinfo{author}{{Google Developers Blog}}, \bibinfo{year}{2023}.
\newblock \bibinfo{title}{Announcing palm api and makersuite}.
\newblock \URLprefix
  \url{https://developers.googleblog.com/en/announcing-palm-api-and-makersuite/}.
\bibitem[{Guidi et~al.(2024)Guidi, Dominici, Gilmour, Butler, Bell, Delaney and
  Bargagli-Stoffi}]{guidi2024environmental}
\bibinfo{author}{Guidi, G.}, \bibinfo{author}{Dominici, F.},
  \bibinfo{author}{Gilmour, J.}, \bibinfo{author}{Butler, K.},
  \bibinfo{author}{Bell, E.}, \bibinfo{author}{Delaney, S.},
  \bibinfo{author}{Bargagli-Stoffi, F.J.}, \bibinfo{year}{2024}.
\newblock \bibinfo{title}{Environmental burden of united states data centers in
  the artificial intelligence era}.
\newblock \bibinfo{journal}{arXiv preprint arXiv:2411.09786} .
\bibitem[{Hirano et~al.(2003)Hirano, Imbens and Ridder}]{hirano2003efficient}
\bibinfo{author}{Hirano, K.}, \bibinfo{author}{Imbens, G.W.},
  \bibinfo{author}{Ridder, G.}, \bibinfo{year}{2003}.
\newblock \bibinfo{title}{Efficient estimation of average treatment effects
  using the estimated propensity score}.
\newblock \bibinfo{journal}{Econometrica} \bibinfo{volume}{71},
  \bibinfo{pages}{1161--1189}.
\newblock \DOIprefix\doi{10.1111/1468-0262.00442}.
\bibitem[{Horsley(2025)}]{horsley2025electricbill}
\bibinfo{author}{Horsley, S.}, \bibinfo{year}{2025}.
\newblock \bibinfo{title}{Is your electric bill going up? ai is partly to
  blame}.
\newblock \bibinfo{journal}{NPR} \URLprefix
  \url{https://www.npr.org/2025/11/06/nx-s1-5597971/electricity-bills-utilities-ai}.
  \bibinfo{note}{published November 6, 2025}.
\bibitem[{IEA(2024)}]{RN1}
\bibinfo{author}{IEA}, \bibinfo{year}{2024}.
\newblock \bibinfo{title}{Energy and AI World Energy Outlook Special Report}.
\newblock \bibinfo{type}{Report}.
\bibitem[{Jacobs(2025)}]{jacobs2025}
\bibinfo{author}{Jacobs, M.}, \bibinfo{year}{2025}.
\newblock \bibinfo{title}{Data centers are already increasing your energy
  bills. we have the receipts.}
\newblock \bibinfo{howpublished}{Union of Concerned Scientists Blog}.
\newblock \URLprefix
  \url{https://blog.ucsusa.org/mike-jacobs/data-centers-are-already-increasing-your-energy-bills/}.
\bibitem[{{KCCI 8 News}(2023)}]{smith2023iowa}
\bibinfo{author}{{KCCI 8 News}}, \bibinfo{year}{2023}.
\newblock \bibinfo{title}{{AI} technology behind {ChatGPT} was built in {Iowa}
  --- with a lot of water}.
\newblock
  \bibinfo{howpublished}{\url{https://www.kcci.com/article/ai-technology-behind-chatgpt-built-in-west-des-moines-iowa-microsoft/45081445}}.
\newblock \bibinfo{note}{Accessed: 2025}.
\bibitem[{Kim et~al.(2025)Kim, Shin, Chung and Rhu}]{kim2025cost}
\bibinfo{author}{Kim, J.}, \bibinfo{author}{Shin, B.}, \bibinfo{author}{Chung,
  J.}, \bibinfo{author}{Rhu, M.}, \bibinfo{year}{2025}.
\newblock \bibinfo{title}{The cost of dynamic reasoning: Demystifying ai agents
  and test-time scaling from an ai infrastructure perspective}.
\newblock \URLprefix \url{https://arxiv.org/abs/2506.04301},
  \href{http://arxiv.org/abs/2506.04301}{\tt arXiv:2506.04301}.
\bibitem[{Klingert et~al.(2018)Klingert, Niedermeier, Dupont, Giuliani, Schulze
  and de~Meer}]{klingert2018mapping}
\bibinfo{author}{Klingert, S.}, \bibinfo{author}{Niedermeier, F.},
  \bibinfo{author}{Dupont, C.}, \bibinfo{author}{Giuliani, G.},
  \bibinfo{author}{Schulze, T.}, \bibinfo{author}{de~Meer, H.},
  \bibinfo{year}{2018}.
\newblock \bibinfo{title}{Mapping data centre business types with power
  management strategies to identify demand response candidates}, in:
  \bibinfo{booktitle}{Proceedings of the 9th International Conference on Future
  Energy Systems}, \bibinfo{publisher}{ACM}, \bibinfo{address}{New York, NY,
  USA}. pp. \bibinfo{pages}{492--503}.
\newblock \DOIprefix\doi{10.1145/3208903.3213521}.
\bibitem[{Laughner et~al.(2024)Laughner, Heckman, Manzur and
  Sloop}]{Laughner2024THD}
\bibinfo{author}{Laughner, T.}, \bibinfo{author}{Heckman, S.},
  \bibinfo{author}{Manzur, A.}, \bibinfo{author}{Sloop, C.},
  \bibinfo{year}{2024}.
\newblock \bibinfo{title}{Analysis of total harmonic distortion on the u.s.
  electric grid}.
\newblock \URLprefix
  \url{https://www.whiskerlabs.com/analysis-of-total-harmonic-distortion-on-the-us-electric-grid/}.
  \bibinfo{note}{accessed January 28, 2026}.
\bibitem[{Lechowicz et~al.(2025)Lechowicz, Comden and
  Bernstein}]{optimizing-grid}
\bibinfo{author}{Lechowicz, A.}, \bibinfo{author}{Comden, J.},
  \bibinfo{author}{Bernstein, A.}, \bibinfo{year}{2025}.
\newblock \bibinfo{title}{Optimizing individualized incentives from grid
  measurements under limited knowledge of agent behavior}, in:
  \bibinfo{booktitle}{Proceedings of the 16th ACM International Conference on
  Future and Sustainable Energy Systems}, \bibinfo{publisher}{Association for
  Computing Machinery}, \bibinfo{address}{New York, NY, USA}. p.
  \bibinfo{pages}{544–563}.
\newblock \URLprefix \url{https://doi.org/10.1145/3679240.3734594},
  \DOIprefix\doi{10.1145/3679240.3734594}.
\bibitem[{Li et~al.(2023)Li, Xie, Dou and Guo}]{li2023gnn}
\bibinfo{author}{Li, Z.}, \bibinfo{author}{Xie, C.}, \bibinfo{author}{Dou, Z.},
  \bibinfo{author}{Guo, M.}, \bibinfo{year}{2023}.
\newblock \bibinfo{title}{A {GNN}-based day ahead carbon intensity forecasting
  model for cross-border power grids}, in: \bibinfo{booktitle}{Proceedings of
  the 14th ACM International Conference on Future Energy Systems},
  \bibinfo{publisher}{ACM}, \bibinfo{address}{New York, NY, USA}. pp.
  \bibinfo{pages}{199--210}.
\newblock \DOIprefix\doi{10.1145/3575813.3597346}.
\bibitem[{Lin et~al.(2024a)Lin, Wijayawardana, Rao, Nguyen, Gnibga and
  Chien}]{lin2024exploding}
\bibinfo{author}{Lin, L.}, \bibinfo{author}{Wijayawardana, R.},
  \bibinfo{author}{Rao, V.}, \bibinfo{author}{Nguyen, H.},
  \bibinfo{author}{Gnibga, E.W.}, \bibinfo{author}{Chien, A.A.},
  \bibinfo{year}{2024}a.
\newblock \bibinfo{title}{Exploding {AI} power use: An opportunity to rethink
  grid planning and management}, in: \bibinfo{booktitle}{Proceedings of the
  15th ACM International Conference on Future and Sustainable Energy Systems},
  \bibinfo{publisher}{ACM}, \bibinfo{address}{New York, NY, USA}. pp.
  \bibinfo{pages}{434--441}.
\newblock \DOIprefix\doi{10.1145/3632775.3661959}.
\bibitem[{Lin et~al.(2024b)Lin, Wijayawardana, Rao, Nguyen, GNIBGA and
  Chien}]{ai-grid-planning}
\bibinfo{author}{Lin, L.}, \bibinfo{author}{Wijayawardana, R.},
  \bibinfo{author}{Rao, V.}, \bibinfo{author}{Nguyen, H.},
  \bibinfo{author}{GNIBGA, E.W.}, \bibinfo{author}{Chien, A.A.},
  \bibinfo{year}{2024}b.
\newblock \bibinfo{title}{Exploding ai power use: an opportunity to rethink
  grid planning and management}, in: \bibinfo{booktitle}{Proceedings of the
  15th ACM International Conference on Future and Sustainable Energy Systems},
  \bibinfo{publisher}{Association for Computing Machinery},
  \bibinfo{address}{New York, NY, USA}. p. \bibinfo{pages}{434–441}.
\newblock \URLprefix \url{https://doi.org/10.1145/3632775.3661959},
  \DOIprefix\doi{10.1145/3632775.3661959}.
\bibitem[{Lindberg et~al.(2021)Lindberg, Abdennadher, Lesieutre and
  Roald}]{lindberg2021guide}
\bibinfo{author}{Lindberg, J.}, \bibinfo{author}{Abdennadher, Y.},
  \bibinfo{author}{Lesieutre, B.C.}, \bibinfo{author}{Roald, L.A.},
  \bibinfo{year}{2021}.
\newblock \bibinfo{title}{A guide to reducing carbon emissions through data
  center geographical load shifting}, in: \bibinfo{booktitle}{Proceedings of
  the 12th ACM International Conference on Future Energy Systems},
  \bibinfo{publisher}{ACM}, \bibinfo{address}{New York, NY, USA}. pp.
  \bibinfo{pages}{157--168}.
\newblock \DOIprefix\doi{10.1145/3447555.3466582}.
\bibitem[{Liu et~al.(2014)Liu, Liu, Low and Wierman}]{liu2014pricing}
\bibinfo{author}{Liu, Z.}, \bibinfo{author}{Liu, I.}, \bibinfo{author}{Low,
  S.}, \bibinfo{author}{Wierman, A.}, \bibinfo{year}{2014}.
\newblock \bibinfo{title}{Pricing data center demand response}, in:
  \bibinfo{booktitle}{Proceedings of the ACM SIGMETRICS/International
  Conference on Measurement and Modeling of Computer Systems},
  \bibinfo{publisher}{ACM}, \bibinfo{address}{New York, NY, USA}. pp.
  \bibinfo{pages}{111--123}.
\newblock \DOIprefix\doi{10.1145/2637364.2592004}.
\bibitem[{Luby et~al.(2025)Luby, Sergi, Azevedo, Davis and Bergman}]{cmu2025}
\bibinfo{author}{Luby, I.}, \bibinfo{author}{Sergi, B.},
  \bibinfo{author}{Azevedo, I.}, \bibinfo{author}{Davis, L.},
  \bibinfo{author}{Bergman, A.}, \bibinfo{year}{2025}.
\newblock \bibinfo{title}{Electricity Grid Impacts of Rising Demand from Data
  Centers and Cryptocurrency Mining Operations}.
\newblock \bibinfo{type}{Technical Report}. Carnegie Mellon University, Scott
  Institute for Energy Innovation.
\newblock \URLprefix
  \url{https://energy.cmu.edu/_files/documents/electricity-grid-impacts-of-rising-demand-from-data-centers-and-cryptocurrency-mining-operations.pdf}.
\bibitem[{MacKay and Mercadal(2024)}]{mackay2024wholesale}
\bibinfo{author}{MacKay, A.}, \bibinfo{author}{Mercadal, I.},
  \bibinfo{year}{2024}.
\newblock \bibinfo{title}{Deregulation, market power, and prices: {Evidence}
  from the electricity sector}.
\newblock \bibinfo{journal}{American Economic Journal: Microeconomics}
  \bibinfo{note}{Finds near one-for-one long-run pass-through of utility costs
  to retail rates; \$14.2/MWh margin increase 2000--2016}.
\bibitem[{Maji et~al.(2022a)Maji, Bashir, Irwin, Shenoy and
  Sitaraman}]{maji2022dacf}
\bibinfo{author}{Maji, D.}, \bibinfo{author}{Bashir, N.},
  \bibinfo{author}{Irwin, D.}, \bibinfo{author}{Shenoy, P.},
  \bibinfo{author}{Sitaraman, R.K.}, \bibinfo{year}{2022}a.
\newblock \bibinfo{title}{{DACF}: Day-ahead carbon intensity forecasting of
  power grids using machine learning}, in: \bibinfo{booktitle}{Proceedings of
  the 13th ACM International Conference on Future Energy Systems},
  \bibinfo{publisher}{ACM}, \bibinfo{address}{New York, NY, USA}. pp.
  \bibinfo{pages}{188--192}.
\newblock \DOIprefix\doi{10.1145/3538637.3538849}.
\bibitem[{Maji et~al.(2024)Maji, Bashir, Irwin, Shenoy and
  Sitaraman}]{maji2024crossroads}
\bibinfo{author}{Maji, D.}, \bibinfo{author}{Bashir, N.},
  \bibinfo{author}{Irwin, D.}, \bibinfo{author}{Shenoy, P.},
  \bibinfo{author}{Sitaraman, R.K.}, \bibinfo{year}{2024}.
\newblock \bibinfo{title}{Data centers carbon emissions at crossroads: An
  empirical study}.
\newblock \bibinfo{journal}{ACM SIGEnergy Energy Informatics Review}
  \bibinfo{volume}{4}.
\newblock \DOIprefix\doi{10.1145/3757892.3757899}.
\bibitem[{Maji et~al.(2025)Maji, Hanafy, Wu, Irwin, Shenoy and
  Sitaraman}]{datacenter-carbon}
\bibinfo{author}{Maji, D.}, \bibinfo{author}{Hanafy, W.A.},
  \bibinfo{author}{Wu, L.}, \bibinfo{author}{Irwin, D.},
  \bibinfo{author}{Shenoy, P.}, \bibinfo{author}{Sitaraman, R.K.},
  \bibinfo{year}{2025}.
\newblock \bibinfo{title}{Data centers carbon emissions at crossroads: An
  empirical study}.
\newblock \bibinfo{journal}{SIGENERGY Energy Inform. Rev.} \bibinfo{volume}{5},
  \bibinfo{pages}{48–55}.
\newblock \URLprefix \url{https://doi.org/10.1145/3757892.3757899},
  \DOIprefix\doi{10.1145/3757892.3757899}.
\bibitem[{Maji et~al.(2022b)Maji, Sitaraman and
  Shenoy}]{carbon-intensity-forecasting}
\bibinfo{author}{Maji, D.}, \bibinfo{author}{Sitaraman, R.K.},
  \bibinfo{author}{Shenoy, P.}, \bibinfo{year}{2022}b.
\newblock \bibinfo{title}{Dacf: day-ahead carbon intensity forecasting of power
  grids using machine learning}, in: \bibinfo{booktitle}{Proceedings of the
  Thirteenth ACM International Conference on Future Energy Systems},
  \bibinfo{publisher}{Association for Computing Machinery},
  \bibinfo{address}{New York, NY, USA}. p. \bibinfo{pages}{188–192}.
\newblock \URLprefix \url{https://doi.org/10.1145/3538637.3538849},
  \DOIprefix\doi{10.1145/3538637.3538849}.
\bibitem[{{Meta AI}(2023)}]{meta_code_llama_paper_2023}
\bibinfo{author}{{Meta AI}}, \bibinfo{year}{2023}.
\newblock \bibinfo{title}{Code llama: Open foundation models for code}.
\newblock \URLprefix
  \url{https://ai.meta.com/research/publications/code-llama-open-foundation-models-for-code/}.
\bibitem[{{Meta Newsroom}(2023)}]{meta_llama2_2023}
\bibinfo{author}{{Meta Newsroom}}, \bibinfo{year}{2023}.
\newblock \bibinfo{title}{Llama 2}.
\newblock \URLprefix \url{https://about.fb.com/news/2023/07/llama-2/}.
\bibitem[{Morrison et~al.(2025)Morrison, Na, Fernandez, Dettmers, Strubell and
  Dodge}]{morrison2025holistically}
\bibinfo{author}{Morrison, J.}, \bibinfo{author}{Na, C.},
  \bibinfo{author}{Fernandez, J.}, \bibinfo{author}{Dettmers, T.},
  \bibinfo{author}{Strubell, E.}, \bibinfo{author}{Dodge, J.},
  \bibinfo{year}{2025}.
\newblock \bibinfo{title}{Holistically evaluating the environmental impact of
  creating language models}, in: \bibinfo{booktitle}{The Thirteenth
  International Conference on Learning Representations (ICLR) Spotlight}.
\newblock \URLprefix \url{https://arxiv.org/abs/2503.05804},
  \href{http://arxiv.org/abs/2503.05804}{\tt arXiv:2503.05804}.
\bibitem[{Murino et~al.(2023)Murino, Monaco, Nielsen, Liu, Esposito and
  Scognamiglio}]{murino2023sustainable}
\bibinfo{author}{Murino, T.}, \bibinfo{author}{Monaco, R.},
  \bibinfo{author}{Nielsen, P.S.}, \bibinfo{author}{Liu, X.},
  \bibinfo{author}{Esposito, G.}, \bibinfo{author}{Scognamiglio, C.},
  \bibinfo{year}{2023}.
\newblock \bibinfo{title}{Sustainable energy data centres: A holistic
  conceptual framework for design and operations}.
\newblock \bibinfo{journal}{Energies} \bibinfo{volume}{16},
  \bibinfo{pages}{5764}.
\bibitem[{{National Electrical Manufacturers Association}(2016)}]{NEMA_MG1}
\bibinfo{author}{{National Electrical Manufacturers Association}},
  \bibinfo{year}{2016}.
\newblock \bibinfo{title}{Ansi/\/nema mg-1-2016: Motors and generators}.
\newblock \URLprefix
  \url{https://www.nema.org/docs/default-source/standards-document-library/mg-1-part-31-watermark.pdf}.
  \bibinfo{note}{section IV Part 31 — Relative equivalent temperature rise:
  10\,\degree C difference shortens insulation life by 50\%}.
\bibitem[{Nicoletti et~al.(2024a)Nicoletti, Malik and
  Tartar}]{Nicoletti_Malik_Tartar_2024}
\bibinfo{author}{Nicoletti, L.}, \bibinfo{author}{Malik, N.},
  \bibinfo{author}{Tartar, A.}, \bibinfo{year}{2024}a.
\newblock \bibinfo{title}{Ai needs so much power, it’s making yours worse}.
\newblock \bibinfo{journal}{Bloomberg} \URLprefix
  \url{https://www.bloomberg.com/graphics/2024-ai-power-home-appliances/}.
  \bibinfo{note}{published December 27, 2024. An exclusive Bloomberg graphics
  story.}
\bibitem[{Nicoletti et~al.(2024b)Nicoletti, Malik and Tartar}]{nicoletti2024}
\bibinfo{author}{Nicoletti, L.}, \bibinfo{author}{Malik, N.},
  \bibinfo{author}{Tartar, A.}, \bibinfo{year}{2024}b.
\newblock \bibinfo{title}{Ai power needs threaten billions in damages for us
  households}.
\newblock \bibinfo{journal}{Bloomberg} \URLprefix
  \url{https://www.bloomberg.com/graphics/2024-ai-power-home-appliances/}.
\bibitem[{{NVIDIA} and {Microsoft}(2021)}]{nvidia2021megatron}
\bibinfo{author}{{NVIDIA}}, \bibinfo{author}{{Microsoft}},
  \bibinfo{year}{2021}.
\newblock \bibinfo{title}{{Megatron-Turing NLG}: The largest and most powerful
  generative language model}.
\newblock \bibinfo{howpublished}{NVIDIA Developer Blog}.
\newblock \bibinfo{note}{Trained on NVIDIA Selene supercomputer}.
\bibitem[{{OpenAI}(2021)}]{openai_api_nowaitlist_2021}
\bibinfo{author}{{OpenAI}}, \bibinfo{year}{2021}.
\newblock \bibinfo{title}{Api no waitlist}.
\newblock \URLprefix \url{https://openai.com/index/api-no-waitlist/}.
\bibitem[{{Pacific Northwest National Laboratory}(2022)}]{pnnl2022}
\bibinfo{author}{{Pacific Northwest National Laboratory}},
  \bibinfo{year}{2022}.
\newblock \bibinfo{title}{Price Formation and Grid Operation Impacts from
  Variable Renewable Energy}.
\newblock \bibinfo{type}{Technical Report} \bibinfo{number}{PNNL-33470}. PNNL.
\newblock \URLprefix
  \url{https://www.pnnl.gov/main/publications/external/technical_reports/PNNL-33470.pdf}.
\bibitem[{Patel et~al.(2024)Patel, Choukse, Zhang, Goiri, Warrier, Mahalingam
  and Bianchini}]{RN4}
\bibinfo{author}{Patel, P.}, \bibinfo{author}{Choukse, E.},
  \bibinfo{author}{Zhang, C.}, \bibinfo{author}{Goiri, {\'I}.},
  \bibinfo{author}{Warrier, B.}, \bibinfo{author}{Mahalingam, N.},
  \bibinfo{author}{Bianchini, R.}, \bibinfo{year}{2024}.
\newblock \bibinfo{title}{Characterizing power management opportunities for
  llms in the cloud}, in: \bibinfo{booktitle}{Proceedings of the 29th ACM
  International Conference on Architectural Support for Programming Languages
  and Operating Systems, Volume 3}, pp. \bibinfo{pages}{207--222}.
\bibitem[{Patterson et~al.(2022)Patterson, Gonzalez, H{\"o}lzle, Le, Liang,
  Munguia, Rothchild, So, Texier and Dean}]{RN8}
\bibinfo{author}{Patterson, D.}, \bibinfo{author}{Gonzalez, J.},
  \bibinfo{author}{H{\"o}lzle, U.}, \bibinfo{author}{Le, Q.},
  \bibinfo{author}{Liang, C.}, \bibinfo{author}{Munguia, L.M.},
  \bibinfo{author}{Rothchild, D.}, \bibinfo{author}{So, D.R.},
  \bibinfo{author}{Texier, M.}, \bibinfo{author}{Dean, J.},
  \bibinfo{year}{2022}.
\newblock \bibinfo{title}{The carbon footprint of machine learning training
  will plateau, then shrink}.
\newblock \bibinfo{journal}{Computer} \bibinfo{volume}{55},
  \bibinfo{pages}{18--28}.
\bibitem[{Patterson et~al.(2021)Patterson, Gonzalez, Le, Liang, Munguia,
  Rothchild, So, Texier and Dean}]{RN2}
\bibinfo{author}{Patterson, D.}, \bibinfo{author}{Gonzalez, J.},
  \bibinfo{author}{Le, Q.}, \bibinfo{author}{Liang, C.},
  \bibinfo{author}{Munguia, L.M.}, \bibinfo{author}{Rothchild, D.},
  \bibinfo{author}{So, D.}, \bibinfo{author}{Texier, M.},
  \bibinfo{author}{Dean, J.}, \bibinfo{year}{2021}.
\newblock \bibinfo{title}{Carbon emissions and large neural network training}.
\newblock \bibinfo{journal}{arXiv preprint arXiv:2104.10350} .
\bibitem[{{PJM Interconnection, L.L.C.}(2025)}]{PJM_DataMiner2_2025}
\bibinfo{author}{{PJM Interconnection, L.L.C.}}, \bibinfo{year}{2025}.
\newblock \bibinfo{title}{Pjm data miner\,2: Energy market generation offers
  (supply stack)}.
\newblock
  \bibinfo{howpublished}{\url{https://dataminer2.pjm.com/feed/energy_market_offers}}.
\newblock \bibinfo{note}{Public dataset providing generator offer curves and
  related market data for the PJM Energy Market. Accessed January 2026}.
\bibitem[{Rosenbaum and Rubin(1983)}]{rosenbaum1983central}
\bibinfo{author}{Rosenbaum, P.R.}, \bibinfo{author}{Rubin, D.B.},
  \bibinfo{year}{1983}.
\newblock \bibinfo{title}{The central role of the propensity score in
  observational studies for causal effects}.
\newblock \bibinfo{journal}{Biometrika} \bibinfo{volume}{70},
  \bibinfo{pages}{41--55}.
\newblock \DOIprefix\doi{10.1093/biomet/70.1.41}.
\bibitem[{Roth et~al.(2023)Roth, Sant'Anna, Bilinski and Poe}]{roth2024did}
\bibinfo{author}{Roth, J.}, \bibinfo{author}{Sant'Anna, P.H.},
  \bibinfo{author}{Bilinski, A.}, \bibinfo{author}{Poe, J.},
  \bibinfo{year}{2023}.
\newblock \bibinfo{title}{What's trending in difference-in-differences? a
  synthesis of the recent econometrics literature}.
\newblock \bibinfo{journal}{Journal of Econometrics} \bibinfo{volume}{235},
  \bibinfo{pages}{2218--2244}.
\newblock \DOIprefix\doi{10.1016/j.jeconom.2023.03.008}.
\bibitem[{Sant'Anna and Zhao(2020)}]{santanna2020doubly}
\bibinfo{author}{Sant'Anna, P.H.C.}, \bibinfo{author}{Zhao, J.},
  \bibinfo{year}{2020}.
\newblock \bibinfo{title}{Doubly robust difference-in-differences estimators}.
\newblock \bibinfo{journal}{Journal of Econometrics} \bibinfo{volume}{219},
  \bibinfo{pages}{101--122}.
\newblock \DOIprefix\doi{10.1016/j.jeconom.2020.06.003}.
\bibitem[{Schwartz et~al.(2020)Schwartz, Dodge, Smith and
  Etzioni}]{schwartz2020green}
\bibinfo{author}{Schwartz, R.}, \bibinfo{author}{Dodge, J.},
  \bibinfo{author}{Smith, N.A.}, \bibinfo{author}{Etzioni, O.},
  \bibinfo{year}{2020}.
\newblock \bibinfo{title}{Green ai}.
\newblock \bibinfo{journal}{Communications of the ACM} \bibinfo{volume}{63},
  \bibinfo{pages}{54--63}.
\bibitem[{Shankar and Reuther(2022)}]{RN9}
\bibinfo{author}{Shankar, S.}, \bibinfo{author}{Reuther, A.},
  \bibinfo{year}{2022}.
\newblock \bibinfo{title}{Trends in energy estimates for computing in
  ai/machine learning accelerators, supercomputers, and compute-intensive
  applications}, in: \bibinfo{booktitle}{2022 IEEE High Performance Extreme
  Computing Conference (HPEC)}, \bibinfo{publisher}{IEEE}. pp.
  \bibinfo{pages}{1--8}.
\bibitem[{Shehabi et~al.(2024)Shehabi, Smith, Sartor, Brown, Herrlin, Koomey,
  Masanet, Horner, Azevedo and Lintner}]{shehabi2024}
\bibinfo{author}{Shehabi, A.}, \bibinfo{author}{Smith, S.},
  \bibinfo{author}{Sartor, D.}, \bibinfo{author}{Brown, R.},
  \bibinfo{author}{Herrlin, M.}, \bibinfo{author}{Koomey, J.},
  \bibinfo{author}{Masanet, E.}, \bibinfo{author}{Horner, N.},
  \bibinfo{author}{Azevedo, I.}, \bibinfo{author}{Lintner, W.},
  \bibinfo{year}{2024}.
\newblock \bibinfo{title}{2024 United States Data Center Energy Usage Report}.
\newblock \bibinfo{type}{Technical Report}. Lawrence Berkeley National
  Laboratory.
\newblock \URLprefix
  \url{https://eta-publications.lbl.gov/sites/default/files/2024-12/lbnl-2024-united-states-data-center-energy-usage-report_1.pdf}.
\bibitem[{{Southwest Power Pool, Inc.\ Market Monitoring
  Unit}(2025)}]{spp2024stateofthemarket}
\bibinfo{author}{{Southwest Power Pool, Inc.\ Market Monitoring Unit}},
  \bibinfo{year}{2025}.
\newblock \bibinfo{title}{2024 Annual State of the Market Report}.
\newblock \bibinfo{type}{Technical Report} \bibinfo{number}{2024 Annual State
  of the Market Report}. Southwest Power Pool, Inc.
\newblock \bibinfo{note}{Published May 28, 2025}.
\bibitem[{Strubell et~al.(2019)Strubell, Ganesh and McCallum}]{strubell19}
\bibinfo{author}{Strubell, E.}, \bibinfo{author}{Ganesh, A.},
  \bibinfo{author}{McCallum, A.}, \bibinfo{year}{2019}.
\newblock \bibinfo{title}{Energy and {Policy} {Considerations} for {Deep}
  {Learning} in {NLP}}, in: \bibinfo{booktitle}{Proceedings of the 57th
  {Annual} {Meeting} of the {Association} for {Computational} {Linguistics}},
  \bibinfo{address}{Florence, Italy}. pp. \bibinfo{pages}{3645--3650}.
\bibitem[{Sun et~al.(2021)Sun, Munro, Kalashnov, Du and
  Wager}]{sun2021treatment}
\bibinfo{author}{Sun, H.}, \bibinfo{author}{Munro, E.},
  \bibinfo{author}{Kalashnov, G.}, \bibinfo{author}{Du, S.},
  \bibinfo{author}{Wager, S.}, \bibinfo{year}{2021}.
\newblock \bibinfo{title}{Treatment allocation under uncertain costs}.
\newblock \bibinfo{journal}{arXiv preprint arXiv:2103.11066} .
\bibitem[{Thangam et~al.(2024)Thangam, Muniraju, Ramesh, Narasimhaiah, Khan,
  Booshan, Booshan, Manickam and Ganesh}]{thangam2024impact}
\bibinfo{author}{Thangam, D.}, \bibinfo{author}{Muniraju, H.},
  \bibinfo{author}{Ramesh, R.}, \bibinfo{author}{Narasimhaiah, R.},
  \bibinfo{author}{Khan, N.M.A.}, \bibinfo{author}{Booshan, S.},
  \bibinfo{author}{Booshan, B.}, \bibinfo{author}{Manickam, T.},
  \bibinfo{author}{Ganesh, R.S.}, \bibinfo{year}{2024}.
\newblock \bibinfo{title}{Impact of data centers on power consumption, climate
  change, and sustainability}, in: \bibinfo{booktitle}{Computational
  Intelligence for Green Cloud Computing and Digital Waste Management}.
  \bibinfo{publisher}{IGI Global Scientific Publishing}, pp.
  \bibinfo{pages}{60--83}.
\bibitem[{{U.S. Fire Administration}(2019)}]{usfa2019electrical}
\bibinfo{author}{{U.S. Fire Administration}}, \bibinfo{year}{2019}.
\newblock \bibinfo{title}{Residential Building Electrical Fires (2014-2016)}.
\newblock \bibinfo{type}{Technical Report} \bibinfo{number}{8}. Federal
  Emergency Management Agency.
\newblock \URLprefix
  \url{https://www.usfa.fema.gov/downloads/pdf/statistics/v19i8.pdf}.
\bibitem[{{Whisker Labs}(2024)}]{WhiskerLabs_THD_2024}
\bibinfo{author}{{Whisker Labs}}, \bibinfo{year}{2024}.
\newblock \bibinfo{title}{Analysis of Total Harmonic Distortion on the U.S.
  Electric Grid}.
\newblock \bibinfo{type}{Technical Report}. Whisker Labs.
\newblock \URLprefix
  \url{https://www.whiskerlabs.com/analysis-of-total-harmonic-distortion-on-the-us-electric-grid/}.
  \bibinfo{note}{data collected from the Ting Sensor Network (Whisker Labs
  power quality dataset)}.
\bibitem[{Wu et~al.(2022)Wu, Raghavendra, Gupta, Acun, Ardalani, Maeng, Chang,
  Aga, Huang and Bai}]{RN7}
\bibinfo{author}{Wu, C.J.}, \bibinfo{author}{Raghavendra, R.},
  \bibinfo{author}{Gupta, U.}, \bibinfo{author}{Acun, B.},
  \bibinfo{author}{Ardalani, N.}, \bibinfo{author}{Maeng, K.},
  \bibinfo{author}{Chang, G.}, \bibinfo{author}{Aga, F.},
  \bibinfo{author}{Huang, J.}, \bibinfo{author}{Bai, C.}, \bibinfo{year}{2022}.
\newblock \bibinfo{title}{Sustainable ai: Environmental implications,
  challenges and opportunities}.
\newblock \bibinfo{journal}{Proceedings of Machine Learning and Systems}
  \bibinfo{volume}{4}, \bibinfo{pages}{795--813}.

\end{thebibliography}

\newpage

\appendix
\section{Appendix}

\subsection{Proofs of Propositions}\label{app:proofs}

The proofs use one elementary identity for the positive part. For any real $x$,
\begin{equation}
    \pos{x}=x+\pos{-x},\qquad \pos{x}\ge 0,\qquad x\mapsto\pos{x}\ \text{is convex and nondecreasing.}
    \label{eq:posident}
\end{equation}
Throughout, the residual is $r=\pos{\ell-m}$ and the wedge is $W=\E[r]$ (Definition~\ref{def:wedge}). We flag explicitly where a claim is an exact consequence of the primitives and where it is a structural construction inherited from the modeling assumptions of Section~\ref{sec:model}; the propositions are analytical statements about the model and are not estimated.

\begin{proof}[Proof of Proposition~\ref{prop:wedge}]
Apply \eqref{eq:posident} with $x=\ell-m$ and take expectations:
\begin{equation*}
    W=\E\!\left[\pos{\ell-m}\right]=\E[\ell-m]+\E\!\left[\pos{m-\ell}\right]=\big(\mu_\ell-\E[m]\big)+\E\!\left[\pos{m-\ell}\right],
\end{equation*}
which is the identity \eqref{eq:wedge-identity}. Under volumetric matching $\E[m]=\mu_\ell$ the first term vanishes, leaving $W^{\mathrm{annual}}=\E[\pos{m-\ell}]\ge 0$ because the integrand is nonnegative. For strictness, suppose $\Pr(m<\ell)>0$. The constraint $\E[m-\ell]=0$ together with $m-\ell$ taking strictly negative values on a set of positive measure forces $m-\ell$ to take strictly positive values on a set of positive measure (otherwise $\E[m-\ell]<0$); hence $\Pr(m>\ell)>0$ and $\E[\pos{m-\ell}]>0$. Given $W>0$, expected residual emissions $e\,\E[r]=e\,W>0$ for $e>0$, and from \eqref{eq:reliability}, $\Phi=\phi_1\E[r]+\phi_2\Var(r)\ge\phi_1 W>0$ for $\phi_1>0$ since the variance term is nonnegative. Monotonicity in misalignment is immediate: holding $\E[m]=\mu_\ell$ fixed, a mean-preserving spread of $m-\ell$ raises $\E[\pos{m-\ell}]$ by convexity of the positive part (the second-order stochastic dominance argument used again in Proposition~\ref{prop:ai-amplifies}).
\end{proof}

\begin{proof}[Proof of Proposition~\ref{prop:entry}]
By \eqref{eq:cost-of-capital}, $\rho(\theta)=\rho_0+\lambda\,\mathrm{CV}(R_\theta)$ is strictly increasing in the coefficient of variation of revenue for $\lambda>0$. Entry built against unbundled certificates earns merchant revenue $R^{m}$ (Assumption~\ref{ass:risk}), so the relevant risk ranking is $\mathrm{CV}(R^{\BTM})<\mathrm{CV}(R^{\PPA})<\mathrm{CV}(R^{m})$, giving
\begin{equation*}
    \rho(\BTM)<\rho(\PPA)<\rho(\REC),\qquad \rho(\REC)=\rho_0+\lambda\,\mathrm{CV}(R^{m}).
\end{equation*}
Hold the per-MW expected-revenue schedule $\bar R(\cdot)$ common across the PPA and BTM regimes---both underwrite the output of a real plant selling into the same energy value---so that the regimes differ only through the financing premium. Since $\bar R'(\cdot)<0$, the schedule is strictly decreasing and its inverse $\bar R^{-1}$ is strictly decreasing. Applying \eqref{eq:entry}, $K_\theta=\bar R^{-1}\!\big(\rho(\theta)\,\kappa\big)$, and the ordering of the arguments $\rho(\BTM)\kappa<\rho(\PPA)\kappa$ reverses to
\begin{equation*}
    K^{\BTM}=\bar R^{-1}\!\big(\rho(\BTM)\kappa\big)>\bar R^{-1}\!\big(\rho(\PPA)\kappa\big)=K^{\PPA}.
\end{equation*}
For unbundled certificates the additionality channel is distinct from the risk channel: the certificate rent accrues to inframarginal, already-built capacity, so the expected revenue schedule \emph{faced by a marginal entrant built for certificates} lies (weakly) below $\rho(\REC)\,\kappa$ over the relevant range, and the zero-profit condition \eqref{eq:entry} is satisfied only at $K^{\REC}\approx 0$. Hence $K^{\BTM}>K^{\PPA}>K^{\REC}\approx 0$. The collapse of $K^{\REC}$ is a consequence of the inframarginality (weak-additionality) assumption recorded in Section~\ref{subsec:instruments}, not of the risk ordering alone.
\end{proof}

\begin{proof}[Proof of Proposition~\ref{prop:ai-amplifies}]
Fix the delivery profile $m$ and write $r=\pos{\ell_0+\Delta-m}$.

\emph{(i) Scale.} For $t\ge 1$, replacing $\Delta$ by $t\Delta\ge\Delta$ (pointwise, since $\Delta\ge 0$) gives $\ell_0+t\Delta-m\ge \ell_0+\Delta-m$ in every state, and because $x\mapsto\pos{x}$ is nondecreasing the residual rises pointwise; taking expectations, $W$ is nondecreasing in the scale of $\Delta$.

\emph{(ii) Mean-preserving spread.} For fixed $m$, the map $\ell\mapsto\pos{\ell-m}$ is convex. If $\Delta'$ is a mean-preserving spread of $\Delta$, then $\ell'=\ell_0+\Delta'$ is a mean-preserving spread of $\ell=\ell_0+\Delta$, and by the Rothschild--Stiglitz characterization the expectation of every convex function weakly increases, so $W'=\E[\pos{\ell'-m}]\ge \E[\pos{\ell-m}]=W$. For the variance term, restrict attention to the net-load region $\{\ell>m\}$, where $r=\ell-m$ moves one-for-one with $\ell$; on that region a spread of $\ell$ transmits directly to $r$, so $\Var(r)$ strictly increases whenever $\Pr(\ell>m)\in(0,1)$. The reliability cost \eqref{eq:reliability} then rises through $\phi_2\Var(r)$, and the price externality rises through the variance of net load entering the affine fossil supply.

\emph{(iii) Alignment.} Decompose the deficit contribution by availability. Writing the residual as $\pos{\ell_0+\Delta-m}$ and noting that delivery $m$ is increasing in availability $a$ under any plant-backed instrument (e.g.\ $m^{\PPA}=\min(K_{\PPA}a,\ell)$), a reduction in $\Cov(\Delta,a)$ shifts the mass of incremental load $\Delta$ toward low-availability (hence low-$m$) states, raising $\E[\pos{\ell-m}]=W$.

\emph{Super-proportionality.} Increasing ``peak intensity'' at fixed mean energy is a mean-preserving spread, which by (ii) and the convexity of the positive part raises grid impact more than a proportional increase in mean load would: the marginal cost of an additional unit of energy delivered as a spike exceeds that of a unit delivered as baseload. This is the analytical content of the claim that spiky AI load is disproportionately costly to the grid.
\end{proof}

\begin{proof}[Proof of Proposition~\ref{prop:colocation}]
Energy adequacy $K^{\BTM}\bar a\ge\mu_\ell$ ensures that total on-site generation over the cycle weakly exceeds total load. If storage $S$ is sufficient to buffer the intra-cycle mismatch---that is, to carry surplus from over-generation states to deficit states subject to the energy-balance constraint of Section~\ref{subsec:model-environment}---then behind-the-meter delivery can be scheduled so that $m^{\BTM}(\omega)\ge\ell(\omega)$ for almost every $\omega$ up to the storage limit. The operator's grid residual is the grid-facing withdrawal $\pos{\ell-m^{\BTM}}$, which is then zero almost surely, so $W^{\BTM}=\E[\pos{\ell-m^{\BTM}}]\to 0$ as storage adequacy is approached (with the round-trip-efficiency correction tightening the required $K^{\BTM}\bar a$ but not altering the limit). With the residual removed, $\E[r]\to 0$ and $\Var(r)\to 0$, so the local reliability contribution $\Phi=\phi_1\E[r]+\phi_2\Var(r)\to 0$.

The sign reversal is a modeling consequence of adding local voltage support, not a further probabilistic fact. Augment the local power-quality response with a support term, $\Phi^{\mathrm{loc}}=\phi_1\E[r]+\phi_2\Var(r)-\psi\,\E[\Delta]$, where $\psi\ge 0$ measures the voltage and ramp support that synchronous or grid-forming on-site capacity supplies by absorbing the spikes $\Delta$ behind the meter. As $\E[r],\Var(r)\to 0$, $\Phi^{\mathrm{loc}}\to-\psi\,\E[\Delta]\le 0$, so the local effect flips from degradation ($\Phi>0$) to neutral-or-improving ($\Phi^{\mathrm{loc}}\le 0$). This is the mechanism whose empirical counterpart is the reversal of the local power-quality coefficient reported in Section~\ref{subsec:power-quality-results}; the model rationalizes the sign change but, as emphasized there, the magnitude is an estimate subject to the stated caveats.
\end{proof}

\begin{proof}[Proof of Proposition~\ref{prop:spatial}]
\emph{(a) Distribution across nodes.} Split an aggregate incremental load $\Delta$ symmetrically across $N$ nodes, $\Delta_i=\Delta/N$, and let $r_i$ denote the resulting local residual and $R=\sum_i r_i$ the aggregate residual, with local reliability cost $\Phi(\cdot)$ convex and $\Phi(0)=0$. Convexity through the origin gives subhomogeneity, $\Phi(r/N)\le \Phi(r)/N$ for $N\ge 1$, hence
\begin{equation*}
    \sum_{i=1}^{N}\Phi(r_i)=N\,\Phi\!\left(\tfrac{R}{N}\right)\le \Phi(R),
\end{equation*}
so distributing the load weakly lowers total reliability cost relative to concentrating it at one node, strictly when $\Phi$ is strictly convex and $R>0$. The same convexity applied to the variance channel shows per-node $\Var(r_i)$ falls as the per-node spike magnitude scales down with $N$.

\emph{(b) Relocation to better-aligned regions.} The wedge $W=\E[\pos{\ell-m}]$ is decreasing in the comovement of delivery and load: raising $\Cov(\ell,a)$ (equivalently, siting where renewable availability better tracks the round-the-clock compute profile) increases the share of load met in-state by plant-backed delivery $m$, lowering $\E[\pos{\ell-m}]$. Independently, the price impact $\%\Delta p=\%\Delta Q^d/\varepsilon^s$ is decreasing in the local supply elasticity $\varepsilon^s$; relocating to a region with greater headroom---a flatter $P(\cdot)$ and larger $\varepsilon^s$---lowers the price externality for a given incremental net load $\%\Delta Q^d$. Both margins reduce grid impact, which is the spatial content tested through the geographic heterogeneity in Section~\ref{Sec::Results}.
\end{proof}

\subsection{Total Harmonic Distortion}
\begin{wrapfigure}{r}{0.3\linewidth}
    \centering
    \vspace{-18pt}
    \includegraphics[width=\linewidth]{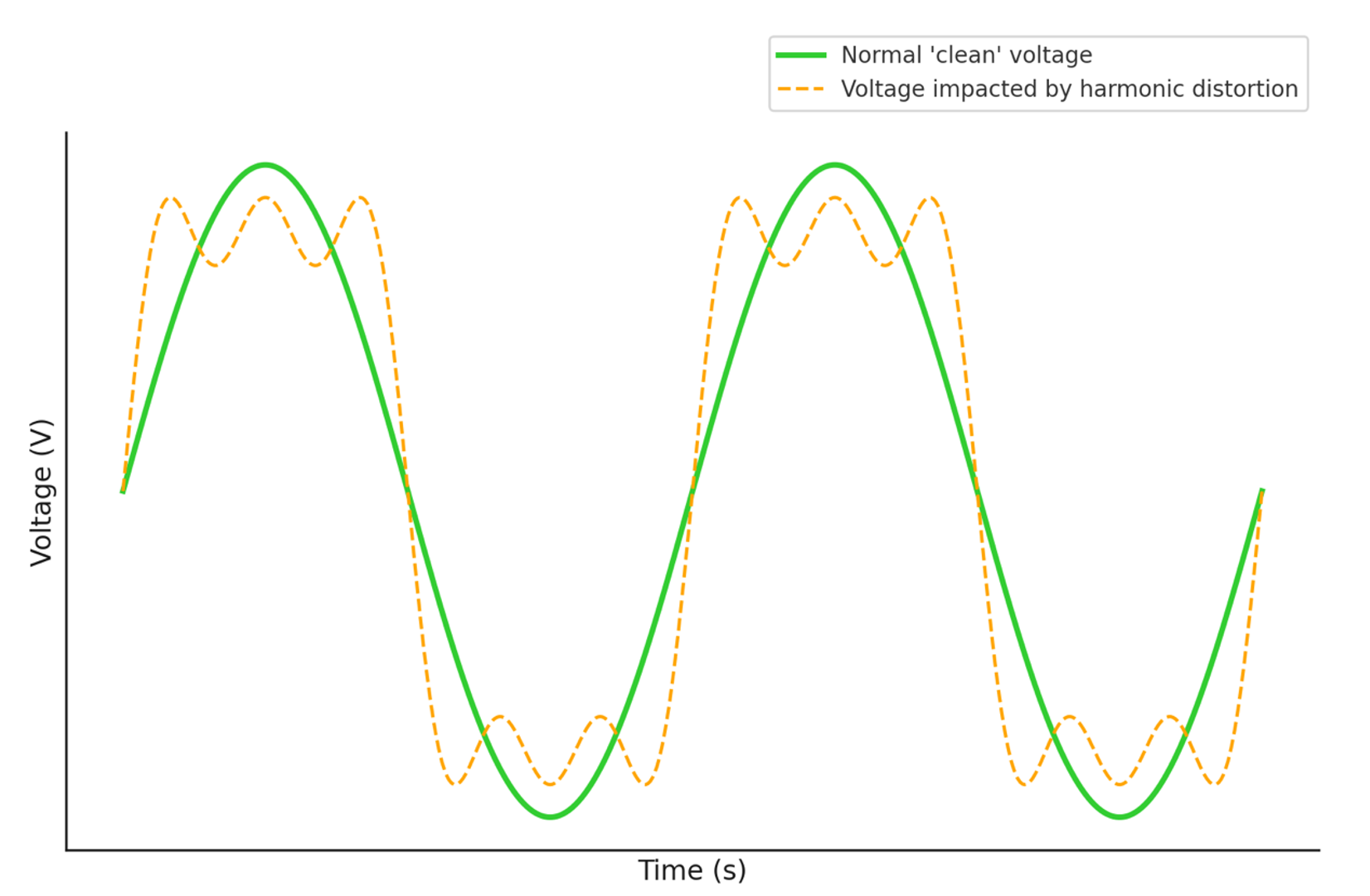}
    \caption{Total harmonic distortion representation.}
    \label{fig:THD}
    \vspace{5pt}
\end{wrapfigure}
Figure~\ref{fig:THD} illustrates the distortions in the power frequencies. The green lie shows the \emph{clear} voltage curve, which is cleanly sinusoidal, whereas with huge loads can produce harmonic distortions shown as the dotted yellow curve. This \emph{dirtier} voltage can adversely affect the operation and lifetime of household appliances and other products that use electricity.

\subsection{Expanded Dataset Description}

Below, maps can be found providing more detail on the data center dataset alongside tables on the datasets and the dataset sample selection criteria.

\begin{figure}[h]
    \centering
    \includegraphics[width=\linewidth]{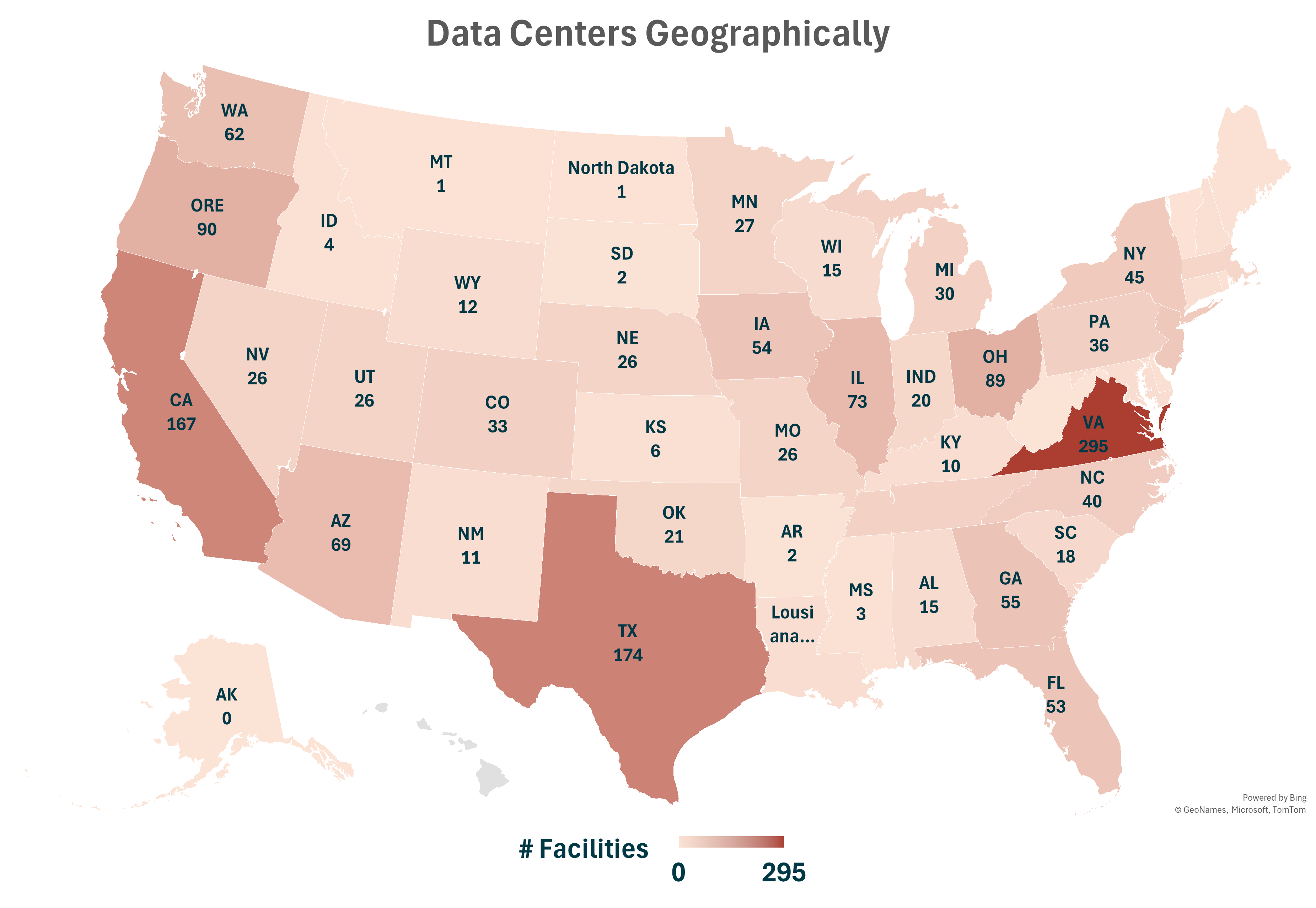}
    \caption{Data centers by state.}
    \label{fig:placeholder}
\end{figure}

\begin{figure}[h]
    \centering
    \includegraphics[width=\linewidth]{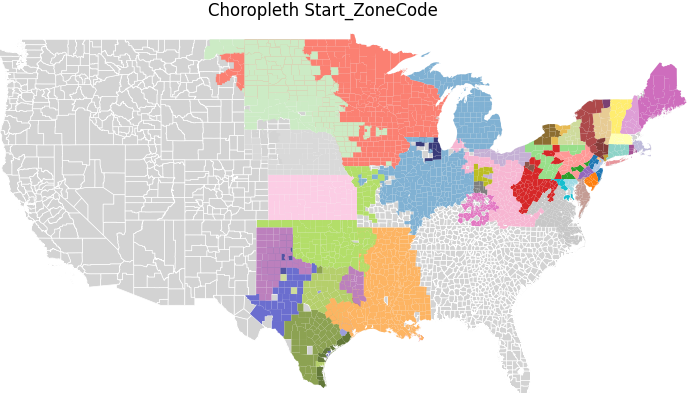}
    \caption{Map of zones for demand model.}
    \label{fig:MapofZonesRegular}
\end{figure}

\begin{figure}[h]
    \centering
    \includegraphics[width=\linewidth]{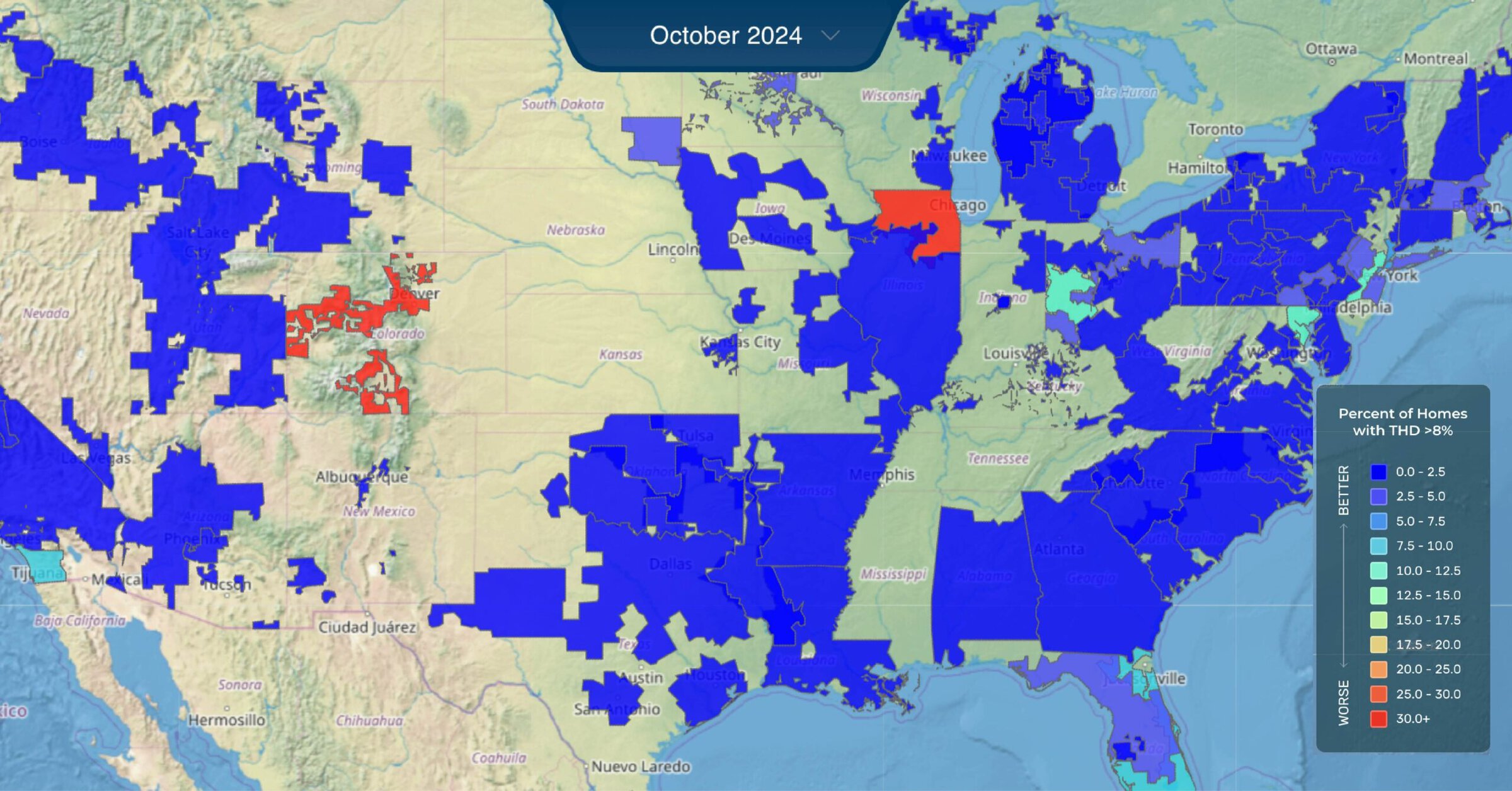}
    \caption{Map of Retail Electric Utilities for Whisker Labs Data.}
    \label{fig:placeholder}
\end{figure}

\begin{figure}
    \centering
    \includegraphics[width=0.75\linewidth]{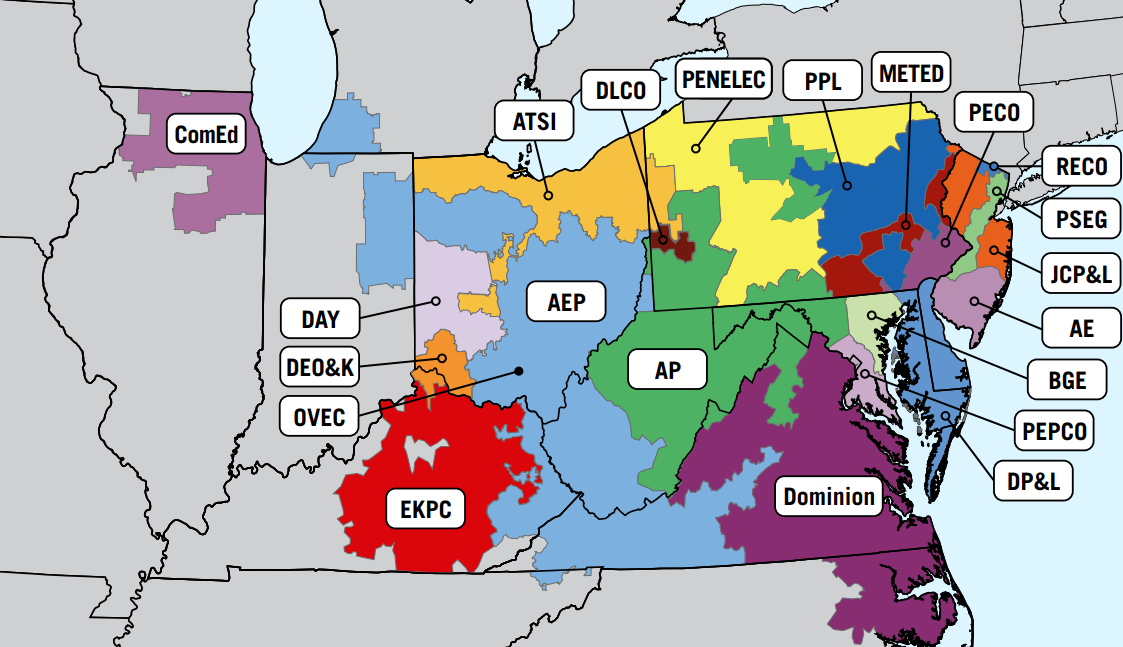}
    \caption{Map of PJM by Zones.}
    \label{fig:placeholder}
\end{figure}

\newpage
\clearpage
\begin{sidewaystable}[p]
\centering
\caption{Expanded description of data sources (for appendix).}
\label{tab:DataSourcesExpanded}

\begin{adjustbox}{max width=\textheight} 
\begin{threeparttable}
\begin{tabular}{p{0.18\textheight} p{0.55\textheight} p{0.25\textheight}}
\toprule
\textbf{Source} & \textbf{Content and Granularity} & \textbf{Use in Analysis} \\
\midrule
\multicolumn{3}{c}{\textbf{Q1: Difference-in-Differences (Power Quality)}} \\
\midrule
Aterio & Data center locations, establishment/expansion/retirement dates, ownership, and operational capacity (MW). Data are available at the facility level and updated annually. & Defines treatment regions; links AI model releases to hyperscaler-owned centers. \\
Whisker Labs & Consumer Power Quality Index (CPQI, composite measure of surges, sags, brownouts, interruptions) and Total Harmonic Distortion (THD, waveform distortion). Provided at county level, hourly frequency. & Outcome variables for DiD regressions. \\
Meteostat & Hourly temperature and precipitation, aggregated at the county level. Coverage includes all U.S. counties with weather stations. & Controls for demand fluctuations and renewable generation variability. \\
ISOs & Geographic boundary shapefiles and wholesale price data at zonal level. Boundaries digitized from ISO-provided maps; prices available hourly. & Market alignment and regional controls. \\
EIA & Retail service territory shapefiles, at utility level, updated periodically. & Used to define treatment/control boundaries and reconcile geographies. \\
\midrule
\multicolumn{3}{c}{\textbf{Q2: Instrumental Variables (Fossil Demand)}} \\
\midrule
EPA CAMPD & Hourly generator-level demand (MWh) and unit-specific heat rates (Btu/kWh). Coverage 2021--2023. & Generator demand is outcome; heat rates serve as instruments for prices. \\
S\&P Capital IQ & Daily natural gas price series, Henry Hub benchmark. National coverage. & Instrument for price and fuel-cost control. \\
Aterio & Generator proximity to data centers owned by releasing companies, matched at 20-mile radius. & Defines treatment generators in IV design. \\
Meteostat & Hourly temperature and precipitation, as above. & Controls for demand and renewable variability. \\
ISOs & Wholesale price series, zonal level, hourly frequency. & Market-level controls. \\
\midrule
\multicolumn{3}{c}{\textbf{Q3: Counterfactual Analyses (Scaling, Efficiency,  \& Tech Development)}} \\
\midrule
Whisker Labs & CPQI and THD (as above). Used to establish baseline deterioration in power quality. & Benchmark outcomes for scaling projections. \\
Epoch & AI model metadata: release dates, FLOPs, parameter counts. Public dataset with model-level detail. & Used for scaling regressions and efficiency counterfactuals. \\
GPU Experiments & Energy use measured on NVIDIA A6000 GPUs. Experiments run with batch size 4, sequence length 1024, and 200 inferences. Each bar in results represents total energy; lighter portion indicates communication overhead. Data collected in-house on 8xA6000 cluster. & Provides empirical GPU energy baselines for counterfactual scaling exercises. \\
\bottomrule
\end{tabular}
\end{threeparttable}
\end{adjustbox}
\end{sidewaystable}
\clearpage

\newpage
\clearpage

\setlength{\rotFPtop}{0pt plus 1fil}
\setlength{\rotFPbot}{0pt plus 1fil}

\begin{sidewaystable}[p] 
\centering
\setlength{\tabcolsep}{3pt}
\renewcommand{\arraystretch}{0.95}

\caption{Dataset sample selection and summary statistics}
\label{tab:SampleSelection}

\begin{threeparttable}
\footnotesize
\begin{tabularx}{\linewidth}{lXl} 
\toprule
\textbf{Dataset} & \textbf{Content / Summary Statistics} & \textbf{Sample Selection} \\
\midrule
Aterio & Data center locations, ownership, and capacity; mapped to ISO zones, EIA service territories, and counties & Hyperscaler data centers (Meta, Microsoft, Amazon, Google; Anthropic via Google) \\
\midrule
Whisker Labs & Reliability measures: CPQI ($N=2{,}683$, mean=0.52, s.d.=0.42); THD (mean=1.81, s.d.=6.77) & 72 retail electric utilities, monthly data 2022--2025 \\
\midrule
CAMPD & Generator-level demand: mean 218 MW (s.d.=158); operating times $\approx 1$ & Several thousand fossil generators, hourly data 2021--2023 \\
\midrule
ISOs & Wholesale electricity prices: mean $\$51$/MWh (s.d.=148); zonal boundaries & Deregulated markets in Eastern Interconnection and ERCOT \\
\midrule
Weather (Meteostat) & Temperature (mean=17°C, s.d.=11.4); precipitation (mean=0.12 mm, s.d.=0.91) & Matched by county/zone, hourly or daily resolution \\
\midrule
External Controls & Natural gas prices, ISO market data, geography/time harmonization & Used for both DiD and IV regressions \\
\bottomrule
\end{tabularx}

\begin{tablenotes}
\scriptsize
\item Notes: Time series are harmonized at hourly or monthly resolution depending on source. Unified panel supports both difference-in-differences (power quality) and instrumental variable (demand) regressions. Maps of sample regions provided in the appendix.
\end{tablenotes}
\end{threeparttable}

\end{sidewaystable}
\clearpage
\newpage

\subsection{Dataset Creation Diagram}

\begin{figure}[htbp]
\centering
\begin{tikzpicture}[
    scale=0.72,
    transform shape,
    node distance=0.48cm and 0.7cm,
    source/.style={
        cylinder,
        shape border rotate=90,
        draw=#1!60!black,
        fill=#1!12,
        text centered,
        minimum height=1.55em,
        minimum width=1.6cm,
        aspect=0.35,
        font=\tiny,
        align=center
    },
    source/.default=green,
    process/.style={
        rectangle,
        draw=blue!50!black,
        fill=blue!10,
        text centered,
        rounded corners=2pt,
        minimum height=1.35em,
        minimum width=1.8cm,
        font=\tiny,
        align=center
    },
    model/.style={
        rectangle,
        draw=#1!60!black,
        fill=#1!12,
        thick,
        text centered,
        rounded corners=3pt,
        minimum height=1.85em,
        minimum width=2.4cm,
        font=\scriptsize\bfseries,
        align=center
    },
    model/.default=red,
    arrow/.style={
        -{Stealth[scale=0.65]},
        semithick,
        color=gray!70!black
    },
    directarrow/.style={
        -{Stealth[scale=0.55]},
        color=blue!50!black,
        thin
    },
    grouplabel/.style={
        font=\tiny\bfseries,
        color=gray!70!black
    },
    annot/.style={
        font=\scriptsize,
        color=gray!60!black
    }
]


\node[source=green!80!black] (aterio) {Aterio};

\node[source=teal, below=0.35cm of aterio] (whisker) {Whisker Labs};

\node[source=blue!70!black, below=0.35cm of whisker] (campd) {EPA CAMPD};
\node[source=blue!70!black, below=0.26cm of campd] (spiq) {S\&P Capital IQ};
\node[source=blue!70!black, below=0.26cm of spiq] (isos) {ISOs};

\node[source=cyan!70!black, below=0.35cm of isos] (eia) {EIA};

\node[source=blue!60!green, below=0.35cm of eia] (weather) {Meteostat};

\node[source=purple, below=0.35cm of weather] (epoch) {Epoch};


\node[process, right=1.5cm of aterio, yshift=-0.3cm] (treat_def) {Treatment\\Definition};

\node[process, right=1.5cm of whisker, yshift=-0.2cm] (outcome_pq) {Power Quality\\Outcomes};

\node[process, right=1.5cm of campd, yshift=-0.3cm] (fossil_dem) {Fossil Demand\\Outcomes};

\node[process, right=1.5cm of isos, yshift=0.1cm] (instruments) {Instruments \&\\Controls};

\node[process, right=1.5cm of eia] (territory) {Service Territory\\Mapping};

\node[process, right=1.5cm of epoch] (scaling) {Scaling\\Regressions};


\node[model=orange, right=3.2cm of outcome_pq, yshift=0.8cm] (q1) {Q1: Difference-in-\\Differences};

\node[model=red, below=0.7cm of q1] (q2) {Q2: Instrumental\\Variables};

\node[model=violet, below=0.7cm of q2] (q3) {Q3: Counterfactual\\Analyses};


\draw[arrow] (aterio.east) -- (treat_def.west);

\draw[arrow] (whisker.east) -- (outcome_pq.west);

\draw[arrow] (campd.east) -- (fossil_dem.west);

\draw[arrow] (spiq.east) -- ++(0.35,0) |- (instruments.170);
\draw[arrow] (isos.east) -- (instruments.west);
\draw[arrow] (weather.east) -- ++(0.35,0) |- (instruments.190);

\draw[arrow] (eia.east) -- (territory.west);

\draw[arrow] (epoch.east) -- (scaling.west);


\draw[arrow] (treat_def.east) -- ++(0.2,0) |- (q1.170);
\draw[arrow] (outcome_pq.east) -- ++(0.3,0) |- (q1.west);
\draw[arrow] (territory.east) -- ++(0.4,0) |- (q1.210);

\draw[arrow] (treat_def.east) -- ++(0.5,0) |- (q2.130);
\draw[arrow] (fossil_dem.east) -- ++(0.2,0) |- (q2.170);
\draw[arrow] (instruments.east) -- ++(0.15,0) |- (q2.west);

\draw[arrow] (outcome_pq.east) -- ++(0.6,0) |- (q3.130);
\draw[arrow] (scaling.east) -- ++(0.2,0) |- (q3.west);

\draw[directarrow] (aterio.east) -- ++(5.3,0) |- (q2.50);
\draw[directarrow] (whisker.east) -- ++(5.5,0) |- (q3.50);

\begin{scope}[on background layer]
    \node[fit=(aterio)(epoch), 
          fill=gray!7, 
          rounded corners=3pt, 
          draw=gray!30,
          inner sep=4pt] (rawbox) {};
    \node[grouplabel, above=0.02cm of rawbox.north] {Data Sources};
    
    \node[fit=(treat_def)(scaling)(territory)(instruments), 
          fill=blue!5, 
          rounded corners=3pt, 
          draw=blue!20,
          inner sep=5pt] (constructbox) {};
    \node[grouplabel, above=0.02cm of constructbox.north] {Variable Construction};
    
    \node[fit=(q1)(q2)(q3), 
          fill=red!5, 
          rounded corners=3pt, 
          draw=red!20,
          inner sep=5pt] (modelbox) {};
    \node[grouplabel, above=0.02cm of modelbox.north] {Estimation};
\end{scope}

\node[font=\fontsize{4}{5}\selectfont, color=gray!50!black, below=0.08cm of rawbox.south, text width=2.6cm, align=center] 
    {3,862 data centers\\22M generator-hours\\2,684 utility-months};

\node[annot, right=0.08cm of q1.east] {\tiny Power Quality};
\node[annot, right=0.08cm of q2.east] {\tiny Fossil Demand};
\node[annot, right=0.08cm of q3.east] {\tiny AI Technical Evolution};

\end{tikzpicture}
\caption{Data sources and analysis pipeline. Data sources flow through variable construction into three estimation strategies addressing power quality impacts (Q1), fossil generation (Q2), and counterfactual scaling scenarios (Q3).}
\label{fig:DataDiagram}
\end{figure}
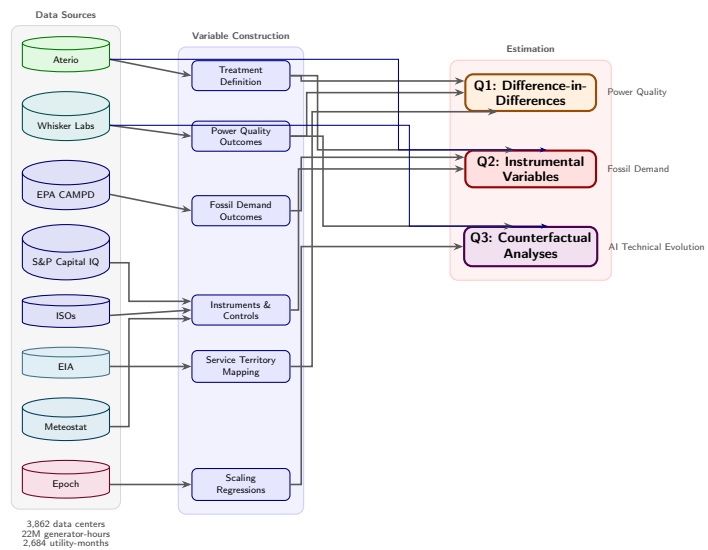

\subsection{Calculation of Comparisons}

We calculate the household average energy usage based on data provided by \cite{USHouseholdConsumption}. We simply divide our estimates by the provided numbers to get the number of households for our comparisons.

\newpage

\subsection{Specifically Verified Models}
\label{appendix:verified-models}
Below is a subset of our dataset for specifically verified models to provide an example of the data:

\begin{table}[htbp]
\centering
\caption{Documented Training Locations (All Confirmed)}
\label{tab:training_locations}
\small
\begin{tabular}{@{}lll@{}}
\toprule
Model & Location & Source \\
\midrule
GPT-4 & W. Des Moines, IA & \cite{smith2023iowa} \\
PaLM, LaMDA, MUM & Mayes County, OK & \cite{googlecloud2023mlhub} \\
Megatron-Turing NLG & Santa Clara, CA & \cite{nvidia2021megatron} \\
\bottomrule
\end{tabular}
\end{table}

\begin{table}[htbp]
\centering
\caption{Verified Model Dates}
\label{tab:verified_model_dates}
\small
\begin{tabular}{@{}llll@{}}
\toprule
Model & Type & Date & Source \\
\midrule
GPT-3 & API & 2021-11-18 & \cite{openai_api_nowaitlist_2021} \\
Claude 2.1 & API & 2023-11-21 & \cite{anthropic_claude2_1_2023} \\
PaLM API & API & 2023-03-14 & \cite{google_palm_api_makersuite_2023} \\
Llama 2 & API & 2023-07-18 & \cite{meta_llama2_2023} \\
Code Llama & Train & 2023-01--07\textsuperscript{*} & \cite{meta_code_llama_paper_2023} \\
\bottomrule
\multicolumn{4}{@{}l@{}}{\footnotesize \textsuperscript{*}Date range (month precision)} \\
\end{tabular}
\end{table}

\subsection{PJM price Elasticities}\label{PriceElasticityAppendix}

\begin{figure}[!htbp]
    \centering
    \includegraphics[width=0.5\linewidth]{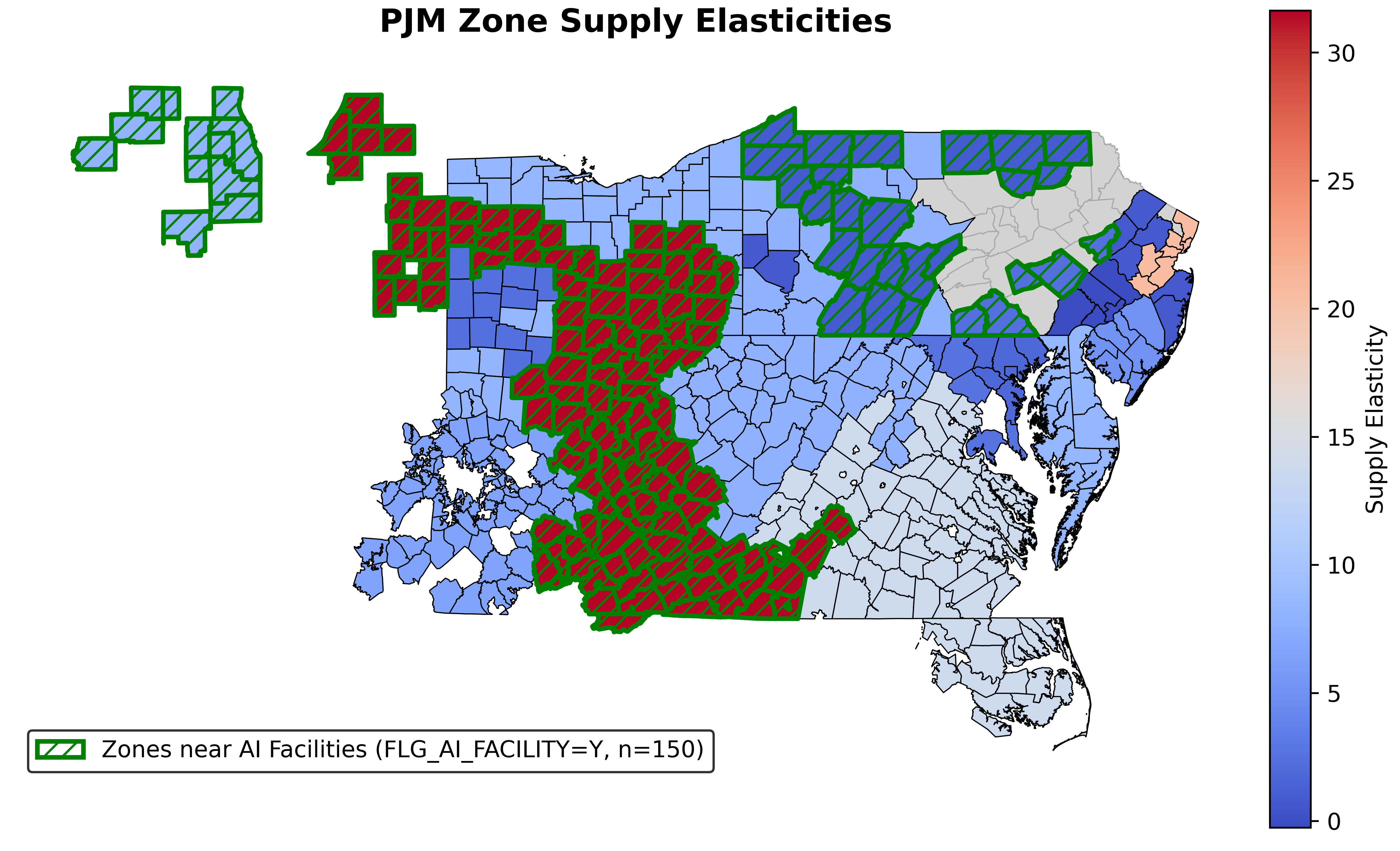}
    \caption{PJM Price Elasticity by Zone}
    \label{PJMPRiceElasticity}
\end{figure}

\subsection{Treatment Definition and Data Linkage}
\label{subsec:treatment}

Our treatment relies on linking AI model releases to geographic regions through a chain of observable connections: models to companies, companies to data center locations, and data center locations to electricity service territories. We describe each link and its limitations.

\subsubsection{Identifying AI Data Centers}

We leverage a recently released variable from Aterio that directly classifies data centers by their primary use case, including a designation for AI-focused facilities. This classification improves upon prior approaches that treated all data centers homogeneously, as AI workloads---particularly large-scale model training---exhibit distinct load profiles characterized by sustained high utilization and substantial power draws that differ markedly from traditional cloud computing or storage operations.

We link data centers to electricity service territories at two geographic resolutions:
\begin{itemize}
    \item \textbf{Utility service territory:} For power quality analysis, we match AI data centers to retail electric utility boundaries. A utility is designated as treated if at least one AI data center operates within its service territory.
    \item \textbf{Generator proximity:} For the fossil generation analysis, we match data centers to generators by picking generators within a 20-mile radius of an in-sample data center. This finer resolution exploits the greater spatial granularity of EPA generator-level data.
\end{itemize}

\subsubsection{Linking Models to Data Centers}

We construct a crosswalk from AI model releases to data center locations through company ownership. For each major model release in our sample, we identify the developing company and link to all AI-classified data centers operated by or contracted to that company. 

\paragraph{Identifying Assumption and Limitations.} Our approach assumes that training and inference activity for a given model occurs at data centers associated with the releasing company. This assumption is imperfect for several reasons:
\begin{enumerate}
    \item Training may occur at third-party cloud facilities not captured in our data center inventory.
    \item Inference workloads may be distributed across facilities differently than training workloads.
\end{enumerate}

We interpret our estimates as capturing the \textit{average} effect across a company's AI data center footprint, which represents a lower bound on the localized effect at the actual training site if activity is concentrated, or an upper bound if we misattribute activity to facilities that were not involved. We probe robustness to this measurement error through several approaches described in Section~\ref{subsec:robustness}.

For a subset of models where public information permits more precise dating, we supplement our treatment timing with verified training and inference windows:

\subsubsection{Instrumental Variables Specification}\label{IVDetails}

The learned DiD coefficients capture the relationship between AI model releases and generator output, but this relationship may be confounded by supply-side shocks. For example, if AI model releases coincide with changes in natural gas prices or generator availability that independently affect electricity production, our estimates would conflate demand-side and supply-side effects. To isolate the demand-side effects, we instrument for electricity prices using the interaction of generator-level heat rates and regional natural gas prices. This instrumentation strategy exploits cross-sectional variation in generators' fuel efficiency: conditional on natural gas price movements, generators with higher heat rates (lower efficiency) face larger marginal cost shocks, allowing us to trace out the demand curve.

\paragraph{First-Stage and Structural Estimates.} The first stage regresses electricity prices on the heat rate--gas price interaction. The structural equation then uses predicted prices to estimate price responsiveness. The estimated price elasticity from the structural equation is 0.135 (SE = 0.013, p < 0.0001), indicating economically meaningful and precisely estimated demand responsiveness. As a specification check, we estimate the first stage separately for each of 5 model-specific regressions; the mean price coefficient is $-0.331$ (SD = 0.005, range: $-0.336$ to $-0.323$). The tight clustering of estimates across specifications suggests stable demand elasticity across AI model events and supports the validity of our instrumentation strategy.

\subsection{Wholesale-to-retail Passthrough Estimates}\label{WholesaleRetailPassthrough}

\begin{table}[!htbp]
\centering
\caption{Wholesale-to-Retail Price Pass-Through in Deregulated Markets}
\label{tab:passthrough}
\scriptsize
\begin{tabular}{@{}llcc@{}}
\toprule
\textbf{Study} & \textbf{ISO} & \textbf{PT} & \textbf{Class} \\
\midrule
\multicolumn{4}{@{}l}{\textit{Panel A: Wholesale $\to$ Retail Pass-Through}} \\
\addlinespace[0.3em]
\citet{defeuilley2023pennsylvania} & PJM (PA) & 0.53 & Res. (EDC) \\
\citet{defeuilley2023pennsylvania} & PJM (PA) & 0.56 & Res. (EGS) \\
\citet{brown2020texas} & ERCOT & 0.43--0.47 & Res. \\
\citet{mackay2024wholesale} & Multi & $\approx$1.0$^b$ & All \\
\midrule
\multicolumn{4}{@{}l}{\textit{Panel B: Demand Elasticity (Residential)}} \\
\addlinespace[0.3em]
\citet{deryugina2020longrun} & PJM (IL) & $-0.09$ & SR (1 mo.) \\
\citet{deryugina2020longrun} & PJM (IL) & $-0.27$ & MR (2 yr.) \\
\citet{deryugina2020longrun} & PJM (IL) & $-0.35$ & LR (10 yr.) \\
\bottomrule
\end{tabular}

\vspace{0.3em}
\raggedright
\scriptsize
\textit{Notes:} PT = pass-through coefficient. EDC = Electric Distribution Company; EGS = Electric Generation Supplier. SR/MR/LR = short/medium/long-run. \\
$^b$Long-run; \$9.59/MWh cost $\to$ \$8.81/MWh retail.
\end{table}


\subsection{Propensity Score Matching}\label{APP:PSM}

Here, we describe the propensity score matching (PSM) and inverse probability weighting (IPW) methodology used to enhance the robustness of our difference-in-differences (DiD) estimates. These methods address potential selection bias arising from non-random placement of AI data centers.

Standard DiD estimation relies on the parallel trends assumption: absent treatment, treated and control units would have followed similar trajectories. However, AI data center locations are chosen strategically based on factors such as electricity prices, grid reliability, and climate conditions. If power plants near AI data centers systematically differ from control plants in observable characteristics correlated with electricity demand trends, the parallel trends assumption may be violated.

Propensity score methods address this concern through three mechanisms: (i) balancing covariates to create comparison groups similar on observable characteristics; (ii) enforcing common support to ensure we only compare units that could plausibly receive either treatment status; and (iii) reducing model dependence by making estimates less sensitive to functional form assumptions \citep{rosenbaum1983central, hirano2003efficient}.

The propensity score is the conditional probability of receiving treatment given observed covariates:
\begin{equation}
    e(\mathbf{X}_i) = \Pr(T_i = 1 \mid \mathbf{X}_i)
\end{equation}
where $T_i$ is a treatment indicator equal to one for observations near an AI data center in the post-treatment period, and $\mathbf{X}_i$ is a vector of pre-treatment covariates.

We estimate propensity scores using logistic regression:
\begin{equation}
    \log\left(\frac{e(\mathbf{X}_i)}{1 - e(\mathbf{X}_i)}\right) = \beta_0 + \boldsymbol{\beta}'\mathbf{X}_i
\end{equation}

The covariates included in the propensity score model are: pre-trend load growth, ambient temperature (°F), dew point (°F), precipitation (inches), plant heat rate (BTU/kWh), and natural gas price (\$/MMBtu). Temperature and dew point capture weather-driven demand variation and cooling load requirements. Precipitation provides additional weather controls. Heat rate proxies for plant efficiency and technology type. Natural gas price captures fuel cost variation affecting dispatch decisions.

We restrict the sample to observations with propensity scores in the common support region $[0.1, 0.9]$:
\begin{equation}
    \mathcal{S} = \{i : 0.1 \leq e(\mathbf{X}_i) \leq 0.9\}
\end{equation}

This trimming rule excludes observations with extreme propensity scores---units almost certain to be treated or untreated---for which comparable counterfactuals do not exist \citep{crump2009dealing}. Observations with $e(\mathbf{X}_i) < 0.1$ have no comparable treated units, while those with $e(\mathbf{X}_i) > 0.9$ have no comparable control units. The trimmed sample ensures that treatment effect estimates rely on regions of the covariate space where both treatment statuses are observed.

IPW creates a pseudo-population where treatment assignment is independent of observed covariates by weighting observations inversely to their probability of receiving their actual treatment status \citep{hirano2003efficient}. We employ stabilized weights to reduce variance:
\begin{equation}
    w_i^{\text{stab}} = \begin{cases}
        \displaystyle\frac{\Pr(T=1)}{e(\mathbf{X}_i)} & \text{if } T_i = 1 \\[10pt]
        \displaystyle\frac{1 - \Pr(T=1)}{1 - e(\mathbf{X}_i)} & \text{if } T_i = 0
    \end{cases}
\end{equation}
where $\Pr(T=1)$ is the unconditional (marginal) treatment probability in the sample.

To prevent extreme weights from inducing instability, we cap weights at the interval $[0.01, 100]$. Within each treatment group, weights are then normalized to sum to the group size:
\begin{equation}
    \tilde{w}_i = w_i^{\text{capped}} \times \frac{n_g}{\sum_{j \in g} w_j^{\text{capped}}}
\end{equation}
where $g \in \{\text{treated}, \text{control}\}$ and $n_g$ is the number of observations in group $g$.

Our baseline DiD specification takes the form:
\begin{equation}
    Y_{it} = \alpha + \beta_1 \text{Post}_t + \beta_2 \text{Treat}_i + \delta (\text{Post}_t \times \text{Treat}_i) + \mathbf{X}_{it}'\boldsymbol{\gamma} + \varepsilon_{it}
\end{equation}
where $Y_{it}$ is the outcome of interest (electricity demand), $\text{Post}_t$ indicates the post-treatment period, $\text{Treat}_i$ indicates treatment group membership, and $\delta$ is the treatment effect of interest.

With IPW, we estimate a weighted least squares version:
\begin{equation}
    \hat{\boldsymbol{\theta}}^{\text{IPW}} = \arg\min_{\boldsymbol{\theta}} \sum_{i,t} \tilde{w}_{it} \left(Y_{it} - \mathbf{Z}_{it}'\boldsymbol{\theta}\right)^2
\end{equation}
where $\mathbf{Z}_{it}$ includes all regressors and $\tilde{w}_{it}$ are the normalized IPW weights.

Our primary specification combines IPW with instrumental variables to address both selection bias (via IPW) and price endogeneity (via IV). The instruments for electricity price are natural gas price and zone-average heat rate. The resulting IV-IPW-DiD estimator is implemented via weighted two-stage least squares, following the doubly robust approach of \citet{santanna2020doubly}.

We assess covariate balance using the standardized mean difference (SMD):
\begin{equation}
    \text{SMD}_k = \frac{\bar{X}_{k,\text{treated}} - \bar{X}_{k,\text{control}}}{\sqrt{(s^2_{k,\text{treated}} + s^2_{k,\text{control}})/2}}
\end{equation}
where $\bar{X}_{k,g}$ and $s^2_{k,g}$ denote the sample mean and variance of covariate $k$ in group $g$.

Following standard practice, we consider $|\text{SMD}| < 0.1$ as indicating good balance, $0.1 \leq |\text{SMD}| < 0.25$ as moderate imbalance warranting caution, and $|\text{SMD}| \geq 0.25$ as substantial imbalance \citep{austin2011introduction}. After weighting, we compute weighted SMDs using IPW weights to verify that the reweighting procedure successfully balances covariates across treatment groups.

We implement several robustness checks to validate our PSM-IPW approach.

First, we compare treatment effect estimates from standard (unweighted) DiD against IPW-DiD. Similar estimates across specifications suggest results are robust to selection on observables.

Second, we test sensitivity to trimming thresholds by varying the propensity score bounds (e.g., $[0.05, 0.95]$, $[0.15, 0.85]$, $[0.20, 0.80]$). Stable estimates across thresholds indicate robustness to the choice of common support region.

Third, we conduct a pre-trends test by estimating a specification that includes a placebo pre-treatment indicator. An insignificant coefficient on this term ($p > 0.05$) supports the parallel trends assumption.

Fourth, we replicate the analysis using only confirmed AI training facility locations from verified sources, reducing potential measurement error in treatment assignment.

Fifth, we compare covariate balance statistics before and after weighting. IPW should substantially reduce SMDs across all covariates.

The IPW-DiD estimator provides a causal estimate of the average treatment effect under the following identifying assumptions: (i) conditional independence---treatment assignment is independent of potential outcomes conditional on $\mathbf{X}$; (ii) parallel trends---absent treatment, treated and control units would follow similar outcome trajectories; (iii) common support---there exists positive probability of treatment for all covariate values; (iv) stable unit treatment value assumption (SUTVA)---no spillovers between units; and (v) correct propensity score specification---the logistic regression model is correctly specified \citep{abadie2005semiparametric, callaway2021difference}.

An important limitation is that propensity score methods only control for selection on \emph{observable} characteristics. If unobserved factors influence both data center location decisions and electricity demand trends, our estimates may still be biased. We partially address this concern through instrumental variables estimation, but acknowledge that unobserved confounding cannot be definitively ruled out in observational settings.

\subsection{Additional Robustness Checks}\label{App:AdditionalRobustnessChecks}


\subsubsection{Robustness Checks}
We test result robustness by varying treated population definitions, estimating specifications with alternative geographic boundaries to confirm consistency. This demonstrates findings reflect genuine causal impacts rather than arbitrary boundary choices. While we lack data on which models run at specific centers, our DiD methodology accounts for this by using publication dates as exogenous treatment, utilizing minimal geographic and temporal treatment windows. We plan robustness checks restricting treatment to larger data centers most likely involved in AI training. To verify treatment effects aren't noise-related, we include appendix robustness checks with randomized treatment dates as additional evidence that effects relate to model releases.

\subsubsection{Treatment Intensity}

Rather than binary treatment indicators, we estimate a continuous treatment intensity specification using total electricity demand (MWh) as the treatment variable. The significant coefficient (p < 0.001; THD: $+0.00018$, p < 0.001) indicate dose-response relationships: larger electricity loads are associated with larger power quality changes.

\subsubsection{Confirmed Locations and Verified Dates}

We re-estimate our main specifications using stricter treatment definitions. The confirmed-location analysis restricts treatment to facilities with precisely geocoded training locations (mean effect: 74.19 MWh). The verified-date analysis uses API release timing rather than publication dates from the Epoch database (mean effect: 63.20 MWh). Both analyses corroborate our main findings with expected attenuation from the more conservative treatment definitions.

\subsubsection{Callaway-Sant'Anna Estimator}

As a robustness check and to provide transparent decomposition of treatment effect heterogeneity, we also implement the \citet{callaway2021difference} estimator. Like the stacked approach, CS avoids using already-treated units as controls; however, rather than achieving this through sample construction, it does so by estimating separate ATT(g,t) parameters for each cohort g and time period t, then aggregating these to form summary measures. This approach explicitly avoids problematic comparisons and provides a framework for testing parallel trends through examination of pre-treatment ATT(g,t) estimates. 
\subsubsection{Time-Series Anomaly Detection}

As an internal consistency check, we apply time-series anomaly detection methods to the outcome variables in treated units. If AI model releases causally affect grid outcomes, we should observe detected anomalies clustering around treatment dates. Specifically:
\begin{enumerate}
    \item For each treated unit, we fit a baseline time-series seasonal ARIMA model to the pre-treatment period.
    \item We generate one-step-ahead forecasts and prediction intervals for the post-treatment period.
    \item We flag observations where realized values fall outside the 20\% prediction interval as anomalies.
    \item We test whether anomaly frequency is elevated in the post-treatment window relative to the pre-treatment window and relative to control units.
\end{enumerate}
This approach provides model-free evidence of structural breaks coinciding with treatment timing.

\subsection{Discussions}

\subsubsection{Limitations}

While informative, this approach inherits the limitations of all Difference-in-Differences experimental designs. It relies on the parallel trends assumption, which may be violated if treated and control regions were already diverging or if generators anticipated model releases.

\subsubsection{Why our Estimates are Higher than Previous Estimates}
Our estimates tend to find larger impacts of AI than previous estimates from the computer science literature \cite{strubell19, schwartz2020green}. We attribute this discrepancy to our choice to utilize empirical data on the realization of outcome variables following AI activity. Compared to traditional approaches like PUE, we assert that our approach is better able to consider the full costs of training and inference including both indirect and direct costs and may be more suited to precise and accurate estimation of AI energy usage.





\end{document}